\documentclass[twoside,11pt]{article}

\usepackage{mathtools}

\usepackage[abbrvbib, preprint]{jmlr2e}
\usepackage{makecell}
\usepackage{bm}
\usepackage{booktabs}
\usepackage{siunitx}
\usepackage{blindtext}
\usepackage{amsmath}
\usepackage{ulem}
\usepackage{statmath}
\usepackage{algorithm}
\usepackage{algpseudocode}

\newtheorem{theorem}{Theorem}[section]
\newtheorem{proposition}[theorem]{Proposition}
\newtheorem{lemma}[theorem]{Lemma}
\newtheorem{corollary}[theorem]{Corollary}
\newtheorem{definition}[theorem]{Definition}
\newtheorem{assumption}[theorem]{Assumption}
\newtheorem{remark}[theorem]{Remark}

\newcommand{\mathbbm}[1]{\text{\usefont{U}{bbm}{m}{n}#1}} 
\newcommand{\Var}{{\mathbf{Var}}}
\newcommand{\R}{\mathbb{R}}

\usepackage{lastpage}
\jmlrheading{23}{2022}{1-\pageref{LastPage}}{1/21; Revised 5/22}{9/22}{21-0000}{Farahi, Rose, and Torrey}

\ShortHeadings{Simulation-Based Inference via Data Projection}{}
\firstpageno{1}

\title{Simulation-Based Inference via Regression Projection and Batched Discrepancies}

\author{\name Arya~Farahi \email arya.farahi@austin.utexas.edu \\
       \addr Departments of Statistics and Data Sciences, University of Texas at Austin, Austin, TX 78712, USA\\
       The NSF-Simons AI Institute for Cosmic Origins, USA
       \AND
       \name Jonah~Rose \email jr8952@princeton.edu \\
       \addr Department of Physics, Princeton University, Princeton, NJ 08544, USA \\
       Center for Computational Astrophysics, Flatiron Institute, New York, NY 10010, USA
       \AND
       \name Paul~Torrey \email paul.torrey@virginia.edu \\
       \addr Department of Astronomy, University of Virginia, Charlottesville, VA 22904, USA \\
       Virginia Institute for Theoretical Astronomy, University of Virginia, Charlottesville, VA 22904, USA\\
       The NSF-Simons AI Institute for Cosmic Origins, USA
       }

\editor{TBD}

\begin{document}

 \maketitle

\begin{abstract}
We analyze a lightweight simulation-based inference method that infers simulator parameters using only a regression-based projection of the observed data. After fitting a surrogate linear regression once, the procedure simulates small batches at proposed parameter values and assigns kernel weights based on the resulting batch residual discrepancy, producing a self-normalized pseudo-posterior that is simple, parallelizable, and requires access only to the fitted regression coefficients rather than raw observations. We formalize the construction as an importance-sampling approximation to a population target that averages over simulator randomness, prove consistency as the number of parameter draws grows, and establish stability to estimating the surrogate regression from finite samples. We then characterize asymptotic concentration as batch size increases and bandwidth shrinks, showing that the pseudo-posterior concentrates on an identified set determined by the chosen projection, thereby clarifying when the method yields point versus set identification. Experiments in a tractable nonlinear model and a cosmological calibration task using the DREAMS simulation suite illustrate the computational advantages of regression-based projections and the identifiability limitations that arise from low-information summaries.

\end{abstract}



\newpage

\section{Introduction}

Complex scientific numerical simulators have become indispensable tools in scientific domains, including cosmology \citep{vogelsberger2020cosmological}, climate \citep{lai2025machine}, epidemiology \citep{heesterbeek2015modeling}, and systems biology \citep{kitano2002computational}, among other fields. These simulators encode mechanistic domain knowledge in the form of stochastic generative models, typically specified via forward procedures for simulating observations given latent parameters, rather than via analytically tractable likelihoods. When the likelihood function $p(y\mid\theta)$ is unavailable or prohibitively expensive to evaluate, standard Bayesian methods based on explicit likelihood computation or gradient-based optimization are no longer directly applicable. This has motivated a broad class of \textit{likelihood-free} or \textit{simulation-based} inference (SBI) techniques that infer the posterior distribution $p(\theta\mid y_{\mathrm{obs}})$ using only simulated data from a forward data generation model \citep[e.g.,][]{beaumont2010approximate,csillery2010approximate,blum2013comparative,bissiri2016general,brehmer2021simulation,zammit2025neural}.

Early work in this area arose under the umbrella of Approximate Bayesian Computation \citep[][]{tavare1997inferring,robert2014bayesian,sisson2018handbook}. Classical ABC methods approximate the posterior by comparing low-dimensional summary statistics of simulated and observed data, accepting or weighting parameter draws based on a discrepancy between $s(y^{\mathrm{sim}})$ and $s(y^{\mathrm{obs}})$. Rejection ABC and its variants replace explicit likelihood evaluations with kernel-based acceptance rules, while more advanced schemes introduce regression adjustment, sequential Monte Carlo, or Markov chain Monte Carlo to improve efficiency \citep{marin2012approximate,lintusaari2017fundamentals}. Despite well-established asymptotic properties of ABC methods \citep{frazier2018asymptotic} and their conceptual simplicity, the choice of informative summary statistics, the curse of dimensionality in the discrepancy space, and the trade-off between computational cost and approximation quality with tolerance parameter make ABC methods challenging to run on large-scale empirical data.

A common theme in classical ABC and in many modern SBI methods is the use of carefully engineered or learned summary statistics (or low-dimensional embeddings) to compress complex simulator output into a form suitable for inference \citep{fearnhead2012constructing,jiang2024embed}. Although the trend in SBI is toward more flexible, often automatically learned representations, the choice between summary and embedding remains a design decision that affects inference quality and robustness \citep{forbes2022summary,huang2023learning}. 
When the full observation $y$ is high-dimensional or costly to store, transmit, or share, practitioners often compress data into lower-dimensional summaries that retain the system’s most salient features. In cosmology, for instance, large-scale structure surveys are routinely distilled into power spectra, correlation functions, or scaling relations among observable quantities \citep{huterer2017dark}. Such summaries can reduce computational and storage costs by several orders of magnitude and can also help address privacy, proprietary, or data-governance constraints. However, poorly chosen summaries may discard information useful for identifying the underlying parameters \citep{forbes2022summary}. Moreover, once a simulator has been calibrated using a fixed set of summaries, the resulting posterior samples are not always easily reusable when alternative summaries are introduced or when new downstream tasks are considered.

Another thread in the SBI literature focuses on amortized neural methods that learn flexible conditional density estimators or surrogate likelihoods directly from simulator output \citep{greenberg2019automatic,cranmer2020frontier,wang2024comprehensive,zammit2025neural}. Approaches such as neural posterior estimation, neural likelihood estimation, and neural ratio estimation use deep neural networks to approximate $p(\theta \mid y)$ or $p(y \mid \theta)$ up to normalization, typically by training on large collections of simulated data \citep{cranmer2015approximating,papamakarios2016fast,lueckmann2019likelihood,hermans2020likelihood}. These methods can scale to high-dimensional observations and enable amortized inference across multiple datasets, but they often require substantial training effort, careful regularization, and access to large volumes of simulations. They also introduce additional modeling assumptions through the choices of prior, network architectures, and loss function, and can yield posterior approximations that are overconfident or poorly calibrated \citep{hermans2021trust}. In practical scientific applications, this can lead to workflows that are computationally intensive and less transparent, making it more difficult to reason about identifiability, uncertainty quantification, posterior concentration, and sensitivity to model misspecification \citep{bharti2025cost}. Moreover, experimental factors such as measurement noise, instrumental effects, calibration errors, and selection biases can cause empirical observations to differ systematically from idealized simulator output, so training neural SBI methods on clean simulations alone may result in biased or miscalibrated inferences when applied to real data.

To mitigate these challenges, it is proposed to use surrogate models, fit directly to observational data, to define a loss function within an SBI framework \citep[see, e.g.,][]{lilie2025dreams, rose2025dreams, rose2025dreams_sims, garcia2025dreams}. These methods are motivated by the observation that many scientific analyses rely on approximately linear relationships between variables, such as scaling relations in astrophysics or dose-response curves in epidemiology. Rather than approximating the full likelihood or posterior, they estimate a pseudo-posterior. They posit a surrogate linear regression of $Y$ on covariates $X$, with regression coefficients $\beta^\circ$ estimated from the empirical data distribution. Batched simulated datasets are then generated from the simulator, the average residual $Y^{\mathrm{sim}} - \beta^{\circ\top}\mathbf{X}^{\mathrm{sim}}$ is computed within each batch, and a weight is assigned to each parameter draw based on the magnitude of this residual mean. The resulting self-normalized measure over the parameter space is referred to as a pseudo-posterior. This resembles \textit{Bayesian indirect likelihood} (and more broadly indirect inference) methods \citep{heggland2004estimating,drovandi2015bayesian}, in which one replaces the intractable likelihood of the generative model with an \textit{auxiliary} model that is easy to fit, and then measures the agreement between the observed and simulated data through the auxiliary fit. The pseudo-posterior is a deliberately minimal construction. It does not attempt to approximate the full data likelihood, nor does it require learning an embedding. Instead, it treats the residual mean
$$
\mu^\circ(\theta)=\mathbb{E}\big[Y-\beta^{\circ\top}\mathbf X\mid\theta\big]
$$
as a one-dimensional discrepancy that measures how well the simulator reproduces the observed linear projection of data. Parameters for which $\mu^\circ(\theta)$ is close to zero are upweighted; those that induce systematic deviations are downweighted. From a practical standpoint, the method offers an appealing combination of properties: it is easy to implement, scales linearly with the number of simulated batches, trivially parallelizable, and requires only a low-dimensional summary of the observed data. This makes it particularly attractive in settings where simulators are expensive and data sharing is constrained, for example, when empirical data are proprietary, sensitive, or subject to strict governance rules. This places the method within the class of likelihood-free methods, but with an emphasis on interpretability, computational simplicity, and minimal data access: once $\beta^\circ$ is known, the raw observed dataset is no longer required for calibration.

While this strategy is intuitive, scalable, and interpretable, and avoids the need to share raw data, formal guarantees are lacking. In particular, it remains unclear under what conditions the underlying parameters are identifiable, whether the resulting procedures are asymptotically consistent, and how inference is affected by surrogate-model misspecification or finite-sample effects. Despite the absence of theoretical guarantees, these methods have gained traction and have recently been adopted in several applied studies \citep{busillo2023casco, lilie2025dreams, rose2025dreams, rose2025dreams_sims, garcia2025dreams, busillo2025casco}\footnote{We note that even though the idea used in these works is similar, the exact method used to compute the posterior/best-fit computationally is different.}.  Building on this growing practical interest, in this paper, we formally investigate this class of methods and establish the proposed method as a simple and principled addition to the likelihood-free inference toolkit. We first formalize the pseudo-posterior construction as a batched importance-sampling scheme and situate it within the broader SBI landscape.
Next, we provide a detailed theoretical analysis of the pseudo-posterior. Under mild moment and continuity assumptions, we show that the Monte Carlo estimator based on batched simulations is consistent for a well-defined population target $\pi_{M,\tau}$, and that $\pi_{M,\tau}$ itself converges to an idealized limiting measure as the batch size $M$ grows and the kernel bandwidth $\tau$ shrinks at an appropriate rate. When the auxiliary coefficient is fixed at its probability limit $\beta^\circ$, we establish concentration of the population pseudo-posterior around an \textit{identified set}
$$
\Theta^\dagger = \arg\min_{\theta\in\Theta}\mu^\circ(\theta)^2,
$$
and prove conditions under which this set contains the true data-generating parameters $\Theta^\star$. These results clarify what the pseudo-posterior is estimating, how simulation noise propagates through the weighting scheme, and how identifiability is governed by the chosen projection. 

The rest of the paper is organized as follows. Section~\ref{sec:setup} introduces the formal problem setup, defines the batched simulation scheme, and describes the kernel weighting that yields the pseudo-posterior. Section~\ref{sec:theory} develops the theoretical foundations, including central limit results for batched residuals, consistency of the empirical measure, and stability with respect to estimation of the auxiliary regression. Section~\ref{sec:concentration-fixed} focuses on identifiability and concentration under a fixed auxiliary coefficient, characterizing the identified set and establishing concentration of both the population and empirical pseudo-posteriors. Section~\ref{sec:example} presents a cosmological application based on the DREAMS data \citep{rose2025introducing}, highlighting the method's behavior in a realistic scientific context. We conclude in Section~\ref{sec:conclusion} with a discussion of practical implications, limitations, and directions for integrating pseudo-posteriors with more general ABC and SBI workflows.

\section{Problem Setup}
\label{sec:setup}

\paragraph{Observed data and augmented predictors.}
Let $\mathcal{D}_{\mathrm{obs}}=\{(x_i^{\mathrm{obs}},y_i^{\mathrm{obs}})\}_{i=1}^n$ denote the observed sample, where $x_i^{\mathrm{obs}}\in\mathbb{R}^d$ and $y_i^{\mathrm{obs}}\in\mathbb{R}$. To encode an intercept, define the augmented predictors
\[
\mathbf{x}_i^{\mathrm{obs}}=(1,\,x_i^{\mathrm{obs}\top})^\top\in\mathbb{R}^{d+1},
\qquad
\mathbf{X}_{\mathrm{obs}}
=
\begin{bmatrix}
\mathbf{x}_1^{\mathrm{obs}\top}\\[-2pt]
\vdots\\[-2pt]
\mathbf{x}_n^{\mathrm{obs}\top}
\end{bmatrix}\in\mathbb{R}^{n\times(d+1)},
\qquad
\mathbf{y}_{\mathrm{obs}}
=
\begin{bmatrix}
y_1^{\mathrm{obs}}\\[-2pt]
\vdots\\[-2pt]
y_n^{\mathrm{obs}}
\end{bmatrix}\in\mathbb{R}^{n}.
\]

\paragraph{Surrogate linear model (OLS on observed data).}
We fit the surrogate model
\begin{equation}
\mathbf{y}_{\mathrm{obs}}=\mathbf{X}_{\mathrm{obs}}\,\beta_{\mathrm{obs}}+\varepsilon,
\qquad
\mathbb{E}[\varepsilon]=0,\ \Var(\varepsilon)=\sigma^2 I_n,
\label{eq:aux}
\end{equation}
and define the least–squares estimator
\begin{equation}
\widehat{\beta}\equiv\widehat{\beta}_n
=
\arg\min_{\beta\in\mathbb{R}^{d+1}}
\|\mathbf{y}_{\mathrm{obs}}-\mathbf{X}_{\mathrm{obs}}\beta\|_2^2
=
\bigl(\mathbf{X}_{\mathrm{obs}}^\top\mathbf{X}_{\mathrm{obs}}\bigr)^{-1}
\mathbf{X}_{\mathrm{obs}}^\top\mathbf{y}_{\mathrm{obs}}
\quad\text{(assuming invertible)}.
\label{eq:beta-hat}
\end{equation}
Let $\beta^\circ=\plim_{n\to\infty}\widehat{\beta}_n$ denote the best linear projection under the observed law $\mathcal{L}_{\mathrm{obs}}(X,Y)$.
Unless stated otherwise, Sections~\ref{sec:setup}--\ref{sec:theory} treat $\widehat{\beta}$ as fixed; Section~\ref{sec:theory-aux} allows $\widehat{\beta}\to\beta^\circ$.  In the astronomy context, for instance, $\widehat{\beta}$ may encode the empirically observed slope and intercept of the luminosity-rotation relation \citep{steinmetz1999cosmological} or galaxy velocity dispersion and hot gas temperature within galaxy clusters \citep{farahi2018xxl}.  Although these regressions are purely phenomenological, ignoring any complex non-linearity or heteroscedasticity in the data, they provide a low-dimensional summary statistic that is easy to compute and interpretable.

\paragraph{Simulator and hierarchical data-generation.}
Let $\theta\in\Theta$ be latent (physical in the context of astronomy and physics) simulator parameters with prior $p(\theta)$; for each $\theta$, the simulator generates
\begin{align}
\theta & \sim p(\theta), \nonumber\\ 
X^{\mathrm{sim}} & \mid\theta \sim p(x\mid\theta),\\
Y^{\mathrm{sim}} & \mid X^{\mathrm{sim}},\theta \sim p(y\mid x,\theta), \nonumber
\label{eq:simulator}
\end{align}
and we augment $\mathbf{X}^{\mathrm{sim}}=(1,\,X^{\mathrm{sim}\top})^\top\in\mathbb{R}^{d+1}$. In a cosmological context, these parameters may include the matter density $\Omega_m$, the amplitude of the initial fluctuations $\sigma_8$, baryonic feedback efficiencies, and other sub-grid parameters controlling star formation and feedback \citep{weinberg2013observational}. Our goal is to identify those $\theta$ for which the simulator reproduces the observed linear trend summarized by $\widehat{\beta}$.

\paragraph{Batched simulation.}
Let $n_\theta\ge1$ denote the number of independent parameter draws used to form the empirical measure. Draw
\begin{equation*}
\theta_1,\dots,\theta_{n_\theta}\stackrel{\text{iid}}{\sim}p(\theta).    
\end{equation*}
For each $j\in\{1,\dots,n_\theta\}$, generate an independent batch
\begin{equation*}
\{(X_{j,m}^{\mathrm{sim}},Y_{j,m}^{\mathrm{sim}})\}_{m=1}^M,
\qquad
X_{j,m}^{\mathrm{sim}}\sim p(x\mid\theta_j),\quad
Y_{j,m}^{\mathrm{sim}}\sim p(y\mid x,\theta_j),    
\end{equation*}
where $M\ge1$ denotes the batch size, number of simulated pairs per parameter value. We set $\mathbf{X}_{j,m}^{\mathrm{sim}}=(1,\,X_{j,m}^{\mathrm{sim}\top})^\top$.

\paragraph{Batch discrepancy (residual mean).}
Given $\widehat{\beta}$, define for each $\theta$ the (population) residual mean and variance
\begin{equation}
\mu(\theta):=\mathbb{E}\!\big[Y-\widehat{\beta}^\top\mathbf{X}\mid\theta\big],
\qquad
v(\theta):=\Var\!\big(Y-\widehat{\beta}^\top\mathbf{X}\mid\theta\big),
\label{eq:mu-v-def}
\end{equation}
where the expectation is under $p(x\mid\theta)p(y\mid x,\theta)$ and $\mathbf{X}=(1,X^\top)^\top$.
The empirical batch discrepancy for draw $\theta_j$ is the batch average residual
\begin{equation}
R_{\theta_j}
:=\frac{1}{M}\sum_{m=1}^M\Big(Y_{j,m}^{\mathrm{sim}}-\widehat{\beta}^\top\mathbf{X}_{j,m}^{\mathrm{sim}}\Big).
\label{eq:Rtheta-def}
\end{equation}

\paragraph{Kernel weighting and pseudo-posterior.}
Fix a bandwidth $\tau>0$ and define unnormalized and normalized weights
\begin{equation}
\widetilde{w}_j=\exp\!\Big(-\tfrac{R_{\theta_j}^2}{2\tau^2}\Big),
\qquad
w_j=\frac{\widetilde{w}_j}{\sum_{k=1}^{n_\theta}\widetilde{w}_k}.
\label{eq:weights}
\end{equation}
The resulting self-normalized empirical measure
\begin{equation}
\Pi_{n_\theta}(\mathrm{d}\theta)=\sum_{j=1}^{n_\theta} w_j\,\delta_{\theta_j}(\mathrm{d}\theta)
\end{equation}
serves as a surrogate posterior over $\theta$. Since likelihood evaluations are absent, the method belongs to the family of likelihood-free or simulation-based inference techniques.  We now establish its theoretical properties.

For population analyses, it is convenient to introduce the \textit{expected population weight}
\begin{equation}
L_M(\theta;\tau):=\mathbb{E}\left[\exp\!\left\{-\frac{R_{M,\theta}^2}{2\tau^2}\right\}\middle|\theta\right],
\qquad
R_{M,\theta}=\frac{1}{M}\sum_{m=1}^{M}\Big(Y_m^{\mathrm{sim}}-\widehat{\beta}^\top\mathbf{X}_m^{\mathrm{sim}}\Big),
\label{eq:L_M-def}
\end{equation}
and the corresponding population target
\begin{equation}
\pi_{M,\tau}(\mathrm{d}\theta)\ \propto\ p(\theta)\,L_M(\theta;\tau)\,\mathrm{d}\theta.
\label{eq:pi-pop}
\end{equation}
Section~\ref{sec:theory} analyzes $\Pi_{n_\theta}$ as an importance-sampling approximation to $\pi_{M,\tau}$, establishes Gaussian control for $R_\theta$, self-normalized consistency as $n_\theta\to\infty$, stability under $\widehat{\beta}\to\beta^\circ$, and concentration as $M\to\infty$ with $\tau\downarrow0$.

Although this methodology bears resemblance to the generalized Bayes method \citep{bissiri2016general} with properly characterized asymptotics \citep{miller2021asymptotic}, it is different in two key respects. First, the update is not of the usual Gibbs and generalized Bayes form
$\pi(\mathrm d\theta\mid y_{\mathrm obs})\propto p(\theta)\exp\{-\lambda \ell(y_{\mathrm obs},\theta)\}\mathrm d\theta$
for a \textit{deterministic} loss evaluated on the observed sample. Instead, our weighting is built from the \textit{random} batch statistic $R_{M,\theta}$, and the population target involves an additional expectation over simulator randomness through $L_M(\theta;\tau)$. Consequently, $\pi_{M,\tau}$ is an $(M,\tau)$-indexed pseudo-posterior rather than the posterior associated with an implicit likelihood for $y_{\mathrm obs}$. Second, the observed data enter the construction only through the plug-in projection $\widehat{\beta}$: rather than matching simulated and observed summaries, the method enforces an approximate moment restriction by upweighting parameters for which the simulated residual mean is close to zero. In the idealized large-batch regime $M\to\infty$, $R_{M,\theta}$ concentrates around $\mu(\theta)$ and $\pi_{M,\tau}$ approaches a Gibbs-type update with loss $\mu(\theta)^2$ and temperature $\tau$ \citep{martin2022direct}; for finite $M$, it can be interpreted as a generalized Bayes update with a noisy loss or, equivalently, as an ABC-style kernel applied to the one-dimensional summary $R_{M,\theta}$ with target value $0$. This distinction clarifies the separate roles of temperature $\tau$ (controlling the strength of the update) and $M$ (simulation noise level), and motivates the joint scaling regimes studied in Section~\ref{sec:theory} and Section~\ref{sec:concentration-fixed}.

We also note that the proposed approach is closely related to \textit{indirect-likelihood} formulations of ABC \citep{drovandi2015bayesian}, in which the intractable likelihood of a mechanistic simulator is replaced by an \textit{auxiliary} statistical model that is straightforward to fit. In such methods, agreement between observed and simulated data is evaluated through the auxiliary fit, yielding an approximate (or ``pseudo-'') posterior that favors parameter values whose simulations reproduce the auxiliary structure present in the empirical distribution. In particular, ABC indirect-score (ABC-IS) and related indirect-inference variants construct discrepancies in terms of auxiliary estimating equations or score functions, rather than raw data-level distances. One treats the regression coefficient vector $\hat\beta_{\mathrm{obs}}$ fitted to the observed data as a low-dimensional summary of the dataset. Parameter values $\theta$ are then assigned higher posterior support when simulations generated at $\theta$ yield auxiliary estimates that are close to their observed counterparts, typically operationalized through kernel weighting or synthetic-likelihood constructions \citep{drovandi2018abc}.

\section{Theoretical Results}
\label{sec:theory}

This section presents the theoretical foundations of the batched weighting scheme.  We present results concerning central limit theorem, consistency, stability, and concentration.  Unless otherwise stated, we treat the surrogate estimate $\widehat{\beta}$ as fixed; Section~\ref{sec:theory-aux} relaxes this by letting $\widehat{\beta}$ be estimated with an estimator converging to its probability limit.

\subsection{Moment Structure and Central Limit Theorem}

For the central limit theorem to apply we impose the following
regularity condition.

\begin{assumption}[Moment regularity, independence, and continuity]
\label{ass:moments}
Fix $\theta\in\Theta$ and define the residual
\begin{equation}
Z = Y - \widehat{\beta}^{\top}\mathbf{X}, \qquad 
(X,Y) \sim p(x\mid\theta)p(y\mid x,\theta), \quad 
\mathbf{X} = (1,X^{\top})^{\top}.
\end{equation}
For each $\theta$:
\begin{enumerate}
\item \textbf{(Independent sampling)} Conditional on $\theta$, the simulated residuals 
$\{Z_m\}_{m=1}^M$ used to form 
$R_{M,\theta} = \frac{1}{M}\sum_{m=1}^{M} Z_m$
are independent and identically distributed, with
\begin{equation}
\mu(\theta) := \mathbb{E}[Z\mid\theta],
\qquad
v(\theta) := \Var(Z\mid\theta)\in[0,\infty).
\end{equation}
\item \textbf{(Finite moments)} The residual has finite fourth moment,
\begin{equation}
\mathbb{E}\bigl[|Z|^4 \mid \theta \bigr] < \infty.
\end{equation}
\item \textbf{(Continuity)} The maps 
$\theta \mapsto \mu(\theta)$ and $\theta \mapsto v(\theta)$ 
are continuous on~$\Theta$.
\end{enumerate}
\end{assumption}

Under Assumption~\ref{ass:moments}, Lyapunov’s condition with exponent~$2$ is satisfied, and therefore the Lindeberg--Feller central limit theorem applies to $R_{M,\theta}$ \citep{billingsley2017probability}. Since $R_{M,\theta}$ is the mean of $M$ i.i.d. residuals, the Lindeberg--Feller central limit theorem yields:

\begin{lemma}
\label{lem:clt}
Let $s^2(\theta) = v(\theta)/M$.  Under Assumption\,
\ref{ass:moments},
\begin{equation}R_{M,\theta} \overset{d}{\longrightarrow} \mathcal{N}\bigl(\mu(\theta), s^2(\theta)\bigr)
\quad\text{as }M\to\infty.
\end{equation}
Moreover, for any bounded continuous function $\varphi$,
\begin{equation}\mathbb{E}[\varphi(R_{M,\theta})\mid\theta] \longrightarrow \int \varphi(t)
\, \mathcal{N}(t\mid\mu(\theta), s^2(\theta))\,\mathrm{d}t.
\end{equation}
\end{lemma}

\subsection{Population Weight Function}
Recall the expected population weight
\begin{equation}
L_M(\theta;\tau) := \mathbb{E}\Bigl[\exp\bigl\{ -R_{M,\theta}^2/(2\tau^2)\bigr\}\mid\theta\Bigl].
\end{equation}
Under the Gaussian approximation of Lemma~\ref{lem:clt}, we may compute $L_M$ explicitly when $R_{M,\theta}$ is normal.  The following lemma uses the moment generating function of a squared Gaussian.

\begin{lemma}
\label{lem:L_M-gauss}
Suppose $R_{M,\theta}\sim\mathcal{N}\bigl(\mu(\theta), s^2(\theta)\bigr)$ with
$s^2(\theta)=v(\theta)/M$ and fix $\tau>0$. Then
\[
L_M(\theta;\tau)
\;=\;
\mathbb{E}\!\left[\exp\!\left\{-\frac{R_{M,\theta}^2}{2\tau^2}\right\}\middle|\theta\right]
\;=\;
\sqrt{\frac{\tau^2}{\tau^2+s^2(\theta)}}
\;\exp\!\left\{-\frac{\mu(\theta)^2}{2\bigl(\tau^2+s^2(\theta)\bigr)}\right\}.
\]
\end{lemma}

\paragraph{Large-$M$ limit.}
As the number of simulated draws per parameter increases, the sampling variance $s^2(\theta)=v(\theta)/M$ of the batch mean $R_{M,\theta}$ vanishes, and the Gaussian approximation in Lemma~\ref{lem:L_M-gauss} becomes increasingly concentrated around its mean $\mu(\theta)$.
Consequently,
\[
L_M(\theta;\tau)
=\sqrt{\frac{\tau^2}{\tau^2+s^2(\theta)}}
\,\exp\!\left\{-\frac{\mu(\theta)^2}{2\bigl(\tau^2+s^2(\theta)\bigr)}\right\}
\; \xlongrightarrow[{M\to\infty}]\;
\exp\!\left\{-\frac{\mu(\theta)^2}{2\tau^2}\right\}
=:L_\infty(\theta;\tau).
\]
Intuitively, as Monte Carlo noise in the batch averages disappears, the population weight depends only on
the deterministic residual mean $\mu(\theta)$ and the kernel bandwidth~$\tau$. This limiting form downweights parameters $\theta$ whose simulated mean residual $\mu(\theta)$ deviates from zero, and therefore quantifies the degree to which the simulator reproduces the observed linear relation summarized by~$\widehat{\beta}$. 

\begin{lemma}[Uniform convergence of weights]
\label{lem:uniform-conv}
Assume $\sup_{\theta\in\Theta}v(\theta)=:V<\infty$.  Then for any fixed $\tau>0$,
\[
\sup_{\theta\in\Theta}\,\bigl|L_M(\theta;\tau)-L_\infty(\theta;\tau)\bigr|\ \longrightarrow\ 0
\qquad\text{as }M\to\infty.
\]
Moreover, the convergence is $\mathcal{O}(1/M)$ uniformly in $\theta$.
\end{lemma}

\subsection{Monte Carlo Consistency}
\label{sec:mc-consistency}

We define the corresponding population pseudo-posteriors as
\begin{equation}
\pi_{M,\tau}(\mathrm{d}\theta)
\ \propto\
p(\theta)\,L_M(\theta;\tau)\,\mathrm{d}\theta,
\qquad
\pi_{\infty,\tau}(\mathrm{d}\theta)
\ \propto\
p(\theta)\,L_\infty(\theta;\tau)\,\mathrm{d}\theta.
\label{eq:pi-pop}
\end{equation}
In the large-$M$ regime, $\pi_{M,\tau}$ thus approximates a smoothed version of the idealized measure $\pi_{\infty,\tau}$, whose concentration around $\mu(\theta)=0$ governs the asymptotic identification of the true parameter set. Here, we demonstrate that our empirical pseudo-posterior $\Pi_{n_\theta}$ approximates the population target $\pi_{M,\tau}$.  We now state a law of large numbers for this importance sampling estimate.

\begin{definition}
For a measurable $h:\Theta\to\mathbb{R}$ with
$\int |h(\theta)|\,\pi_{M,\tau}(\mathrm d\theta)<\infty$ and 
$\int |h(\theta)|\,\pi_{\infty,\tau}(\mathrm d\theta)<\infty$, define
\[
\Phi_M(h)\ :=\ \int h(\theta)\,\pi_{M,\tau}(\mathrm d\theta),
\qquad
\Phi_\infty(h)\ :=\ \int h(\theta)\,\pi_{\infty,\tau}(\mathrm d\theta).
\]
\end{definition}

\begin{theorem}[Monte Carlo consistency]
\label{thm:sn}
Let $h:\Theta\to\mathbb{R}$ be measurable with $\int |h(\theta)|\,\pi_{M,\tau}(\mathrm{d}\theta)<\infty$.  Suppose $\mathbb{E}[L_M(\Theta;\tau)]>0$ for $\Theta\sim p(\theta)$ and  $\mathbb{E}[L_M(\Theta;\tau)\,|h(\Theta)|]<\infty$.  Let $\theta_1,\dots,\theta_{n_\theta}\stackrel{\text{iid}}{\sim}p(\theta)$ and define the self-normalized weights $w_j=\widetilde w_j/\sum_{k=1}^{n_\theta}\widetilde w_k$ with
$\widetilde w_j=\exp\!\bigl(-R_{\theta_j}^2/(2\tau^2)\bigr)$. Then
\[
\sum_{j=1}^{n_\theta} w_j\,h(\theta_j)\ 
\xrightarrow[n_\theta\to\infty]{\text{a.s.}}\ 
\int h(\theta)\,\pi_{M,\tau}(\mathrm{d}\theta).
\]
\end{theorem}

An advantage of the proposed method is the ease of reuse for posterior evaluation. Once the weighted sample $\{(\theta_j,w_j)\}$ has been constructed, it functions as a surrogate posterior: any functional $\int h(\theta)\,\pi(\mathrm d\theta)$ can be approximated by $\sum_j w_j h(\theta_j)$, and any new predictive observable may be evaluated by simulating forward under each $\theta_j$. The calibration need not be repeated when introducing new summary statistics or predictive quantities. In contrast, many likelihood-free approaches must rerun the calibration step when new summaries are introduced.

\begin{proposition} \label{prop:no-unbiasedness}
Unless $L_M(\theta;\tau)=L_\infty(\theta;\tau)$ for $p$-almost every $\theta$ (e.g., when $v(\theta)\equiv 0$ or $\tau\to\infty$), there exist bounded measurable $h$ such that $\Phi_M(h)\neq \Phi_\infty(h)$. 
\end{proposition}

\begin{proof}
If $L_M\not\equiv L_\infty$ on a set of positive $p$-measure, then their induced probability measures $\pi_{M,\tau}$ and $\pi_{\infty,\tau}$ differ. By the definition of total variation, there exists a measurable set $A$ with $\pi_{M,\tau}(A)\neq \pi_{\infty,\tau}(A)$. Taking $h=\mathbbm{1}_A$ (bounded and measurable) gives $\Phi_M(h)\neq \Phi_\infty(h)$.
\end{proof}

This implies that $\Phi_M$ is not an unbiased proxy for $\Phi_\infty$ at finite $M$ in general, but $\Phi_M$ converges to $\Phi_\infty$ as $M\to\infty$.

\begin{theorem}
\label{thm:PhiM-to-PhiInf}
Assume $\sup_{\theta\in\Theta} v(\theta)<\infty$, fix $\tau>0$, and suppose
$\E[\,|h(\Theta)|\,L_\infty(\Theta;\tau)\,]<\infty$ for $\Theta\sim p$.
Then $\Phi_M(h)\to \Phi_\infty(h)$ as $M\to\infty$.
Moreover, if $\sup_{\theta}|L_M(\theta;\tau)-L_\infty(\theta;\tau)|=\mathcal{O}(M^{-1})$, as in Lemma~\ref{lem:uniform-conv}, then for bounded $h$, $|h|\le K<\infty$ for all $\theta \in \Theta$, we have
\begin{equation*}
|\Phi_M(h)-\Phi_\infty(h)|\ =\ \mathcal{O}\!\left(\frac{1}{M}\right).    
\end{equation*}
\end{theorem}

\begin{remark} \label{rem:rate-trunc}
If $\int |h|\,L_\infty\,{\rm d}p<\infty$ and $\int |h|\,|L_M-L_\infty|\,{\rm d}p \le C\,\|L_M-L_\infty\|_\infty$ with finite $C$, then the $\mathcal{O}(M^{-1})$ rate holds. 
\end{remark}

the same rate holds up to an arbitrarily small tail error
via truncation (see Remark~\ref{rem:rate-trunc}).

\begin{corollary}[Two-stage convergence]
\label{cor:two-stage}
Let $\widehat\Phi_{n_\theta,M}(h)=\sum_{j=1}^{n_\theta} w_j h(\theta_j)$ be the self-normalized estimator
with weights as in Equation~\eqref{eq:weights}. 
Under the integrability assumptions of Theorem~\ref{thm:sn} and Theorem~\ref{thm:PhiM-to-PhiInf},
\[
\widehat\Phi_{n_\theta,M}(h)\ \xrightarrow[n_\theta\to\infty]{\mathrm{a.s.}}\ \Phi_M(h)
\quad\text{for each fixed }M,
\qquad\text{and}\qquad
\Phi_M(h)\ \xrightarrow[M\to\infty]{}\ \Phi_\infty(h).
\]
Hence, if $n_\theta\to\infty$ and $M\to\infty$, then 
$\widehat\Phi_{n_\theta,M}(h)\to \Phi_\infty(h)$.
\end{corollary}

\begin{remark}
Self-normalized importance sampling is generally biased for $\Phi_M(h)$ at finite $n_\theta$, though it is strongly consistent (Theorem~\ref{thm:sn}).  Likewise, as one might expect, $\Phi_M(h)$ is generally not equal to $\Phi_\infty(h)$ at finite $M$ (Proposition~\ref{prop:no-unbiasedness}). While unbiasedness w.r.t.\ $\pi_{\infty,\tau}$ does not hold at finite $(n_\theta,M)$; the asymptotic unbiasedness established in Theorem~\ref{thm:sn} and Proposition~\ref{prop:no-unbiasedness}.
\end{remark}

\subsection{Stability Guarantees}
\label{sec:theory-aux}

The preceding results treat $\widehat{\beta}$ as fixed.  In practice, $\widehat{\beta}$ is estimated from the observed data via Equation~\eqref{eq:aux}.  Let $\beta^\circ$ denote the probability limit of $\widehat{\beta}$ as observational sample size $n\to\infty$ (the best linear projection of $Y$ on $\mathbf{X}$ under the data generating process).  Define $\mu^\circ(\theta) = \mathbb{E}[Y - \beta^{\circ\,\top}\mathbf{X}\mid\theta], \quad v^\circ(\theta) = \Var(Y - \beta^{\circ\,\top}\mathbf{X}\mid\theta)$.

\begin{assumption}[Continuity and envelope]
\label{ass:stability}
There exists a neighborhood $\mathcal{N}$ of $\beta^\circ$ such that
 for all $\beta\in\mathcal{N}$ and $\theta\in\Theta$,
\begin{enumerate}
 \item $\mu_{\beta}(\theta) := \mathbb{E}[Y - \beta^\top\mathbf{X}\mid\theta]$ and
 $v_{\beta}(\theta) := \Var(Y - \beta^\top\mathbf{X}\mid\theta)$ are jointly
 continuous in $(\beta,\theta)$, and
 \item there exists an integrable envelope $g$ such that
 $L_M^{(\beta)}(\theta;\tau) \le g(\theta)$ for all $\beta\in\mathcal{N}$,
 where $L_M^{(\beta)}(\theta;\tau)$ denotes $L_M$ with $\widehat{\beta}$
 replaced by $\beta$.
\end{enumerate}
\end{assumption}

\begin{theorem}[Stability under estimated surrogate parameters]
\label{thm:aux}
Under Assumptions~\ref{ass:moments} and~\ref{ass:stability}, if
$\widehat{\beta}\xrightarrow{P}\beta^\circ$, then for any bounded
continuous $h:\Theta\to\mathbb{R}$,
\[
\int h(\theta)\,\pi_{M,\tau}^{(\widehat{\beta})}(\mathrm{d}\theta)
\ \xrightarrow{P}\
\int h(\theta)\,\pi_{M,\tau}^{(\beta^\circ)}(\mathrm{d}\theta),
\]
where $\pi_{M,\tau}^{(\beta)}(\mathrm d\theta)\propto p(\theta)\,L_M^{(\beta)}(\theta;\tau)\,\mathrm d\theta$ and
$L_M^{(\beta)}(\theta;\tau):=\mathbb{E}\!\left[\exp\!\left\{-\tfrac{R_{M,\theta}^{(\beta)\,2}}{2\tau^2}\right\}\middle|\theta\right]$
with $R_{M,\theta}^{(\beta)}=M^{-1}\sum_{m=1}^{M}\big(Y_m-\beta^\top \mathbf X_m\big)$.
\end{theorem}

Theorem~\ref{thm:aux} establishes that the limiting behavior of the pseudo-posterior is stable with respect to estimation error in the surrogate parameter $\widehat{\beta}$.  Even though $\widehat{\beta}$ is estimated from finite data, as long as it converges in probability to its probability limit $\beta^\circ$ and the moment functions vary continuously in $(\beta,\theta)$, the resulting posterior measure over $\theta$ converges, in probability, to the same limit that would be obtained if the true surrogate parameter $\beta^\circ$ were known.  This result guarantees that the uncertainty induced by estimating $\widehat{\beta}$ from the observed data does not propagate asymptotically into the inferential layer over $\theta$: the posterior mass over $\Theta$ remains asymptotically invariant to small perturbations in the surrogate fit.

Combining Theorems~\ref{thm:sn} and~\ref{thm:aux} yields joint convergence when both $n_\theta\to\infty$ and the observed sample size $n\to\infty$. This guarantees that the empirical pseudo-posterior $\Pi_{n_\theta}$ approximates the idealized target $\pi_{M,\tau}^{(\beta^\circ)}$. Consistency in the surrogate regression suffices for stability of the overall inferential procedure, which validates the two-stage estimation framework underlying the pseudo-posterior approach.

\section{Identifiability and Concentration Properties}
\label{sec:concentration-fixed}

So far, we developed the pseudo-posterior framework based on simulation batches and kernel weighting, established its convergence for fixed surrogate parameters, and proved that the method remains stable when the surrogate fit $\beta^\circ$ is replaced by its estimate $\widehat{\beta}$.  We now focus on the population properties of the resulting measure when the surrogate coefficient is fixed at this limit.  In this setting, the simulator parameters $\theta$ are evaluated through their implied residual mean and variance under the best linear projection $\beta^\circ$, yielding the quantities $\mu^\circ(\theta)$ and $v^\circ(\theta)$.  The main goal of this section is to establish concentration of the pseudo-posterior as both the number of simulation draws and the batch size increase. These results provide the formal statistical justification that the surrogate posterior $\Pi_{n_\theta}$ approximates its population counterpart $\pi_{M,\tau}$, and that $\pi_{M,\tau}$ itself concentrates around the identified set $\Theta^\dagger$ as the simulation noise and $\beta$ uncertainties vanish. We show that, under identifiability assumptions, the method recovers the correct identified parameter region in the large-sample limit.

Throughout this section, the surrogate coefficient is fixed at its probability limit $\beta^\circ$; all objects depending on $\beta$ are evaluated at $\beta^\circ$, and we suppress the superscript when unambiguous. Recall
\begin{equation}
\mu^\circ(\theta)=\mathbb{E}\!\big[Y-\beta^{\circ\top}\! \mathbf X \mid \theta\big], 
\qquad
v^\circ(\theta)=\Var\!\big(Y-\beta^{\circ\top}\! \mathbf X \mid \theta\big),
\end{equation}
and, for batch size $M$ and bandwidth $\tau>0$,
\begin{equation}
L_M(\theta;\tau)\ :=\ \mathbb{E}\!\left[\exp\!\left\{-\frac{R_{M,\theta}^2}{2\tau^2}\right\}\middle|\theta\right],
\quad
R_{M,\theta} = \frac1M \sum_{m=1}^M \bigl(Y_m^{\mathrm{sim}}-\beta^{\circ\top}\mathbf X_m^{\mathrm{sim}}\bigr).
\end{equation}
The population target is
\begin{equation}
\pi_{M,\tau}(\mathrm d\theta)\ \propto\ p(\theta)L_M(\theta;\tau)\,\mathrm d\theta,
\end{equation}
and the self-normalized Monte Carlo approximation with i.i.d.\ draws $\theta_1,\dots,\theta_{n_\theta}\sim p$,
\begin{equation}
\Pi_{n_\theta}(\mathrm d\theta)\ =\ \sum_{j=1}^{n_\theta} w_j \,\delta_{\theta_j}(\mathrm d\theta),
\qquad
w_j=\frac{\tilde w_j}{\sum_{k=1}^{n_\theta}\tilde w_k},\quad
\tilde w_j=\exp\!\left\{-\frac{R_{\theta_j}^2}{2\tau^2}\right\}.
\end{equation}

Let the identified set be 
\[
\Theta^\dagger=\arg\min_{\theta\in\Theta}\mu^\circ(\theta)^2
\]
defined as the collection of parameter values that minimize the population discrepancy  $\mu^\circ(\theta)^2$.  Continuity of  $\mu^\circ(\theta)$ on the compact set $\Theta$ guarantees existence but not uniqueness of the minimum.  Hence $\Theta^\dagger$ may be a low-dimensional manifold where dimensionality larger than zero is possible.  In such cases, the underlying simulator and surrogate model are only partially identifying with respect to the observed moment condition, leading to \textit{set identification} rather than point identification. The subsequent analysis assumes that the identified set $\Theta^\dagger$ is a compact set, and all asymptotic concentration results are stated in terms of convergence of probability mass toward this identified set.

\begin{assumption}[Uniform moment regularity]
\label{ass:uniform-moments-fixed}
$\Theta$ is compact and there exist constants $C<\infty$, $\delta>0$, and $0< v_{\min}\le v_{\max}<\infty$ such that for all $\theta\in\Theta$,
\begin{equation}
\sup_{\theta\in\Theta}\mathbb{E}\!\left[|Z|^{2+\delta}\mid\theta\right]\le C,\quad
v^\circ(\theta)\in[v_{\min},v_{\max}],
\end{equation}
where $Z=Y-\beta^{\circ\top}\mathbf X$ and $(X,Y)\sim p(x\!\mid\!\theta)p(y\!\mid\!x,\theta)$. Moreover $\mu^\circ(\theta)$ and $v^\circ(\theta)$ are continuous on $\Theta$.
\end{assumption}

\begin{assumption}[Exact representability] \label{ass:exact_repres}
There exists $\theta^\star\in\Theta$ such that 
$\mathcal L_{\mathrm{sim}}(X,Y\mid\theta^\star)=\mathcal L_{\mathrm{obs}}(X,Y)$, 
where $\mathcal L_{\mathrm{obs}}(X,Y)$ and $\mathcal L_{\mathrm{sim}}(X,Y\mid\theta)$ are the observed law and data generation law, respectively. 
\end{assumption}

Assumption~\ref{ass:exact_repres} implies that the data generation model is correctly specified. So far, we have not made this assumption, but this assumption is made for the concentration guarantees. 

\begin{assumption}[Identifiability]
\label{ass:ident-fixed}
Let $\Theta^\dagger=\arg\min_{\theta\in\Theta}\mu^\circ(\theta)^2$, which is nonempty and compact. Let the prior $p(\theta)$ be continuous and strictly positive on a neighborhood of $\Theta^\dagger$, and for every open $U\supset \Theta^\dagger$ there exists a gap $\eta(U)>0$ such that
\begin{equation}
\inf_{\theta\notin U}\ \mu^\circ(\theta)^2 \ \ge\ \inf_{\vartheta\in\Theta}\ \mu^\circ(\vartheta)^2 + \eta(U).
\end{equation}
\end{assumption}

\begin{lemma}
\label{lem:inf-zero}
Under Assumption~\ref{ass:exact_repres}, we have
\[
\inf_{\vartheta\in\Theta}\ \mu^\circ(\vartheta)^2\;=\;0
\qquad\text{and}\qquad
\theta^\star\in \Theta^\dagger.
\]
\end{lemma}

This is a straightforward consequence of the fact that $\beta^{\circ\top}$ of the OLS estimate and $\mathbb{E}_{\mathrm{obs}}\!\big[Y-\beta^{\circ\top}\mathbf X\big]=0$. Though the inverse is not true necessarily, $\mu^\circ(\theta^\triangle)^2 = 0$. does not imply $\mathcal{L}_{\mathrm{sim}}(X,Y\mid\theta^\triangle)=\mathcal{L}_{\mathrm{obs}}(X,Y)$.

\begin{lemma}
\label{lem:be}
Under Assumption~\ref{ass:uniform-moments-fixed} and and let
\begin{equation}
\Phi(t) \;=\; \int_{-\infty}^t \frac{1}{\sqrt{2\pi}}\,e^{-u^2/2}\,\mathrm du,
\end{equation} 
then
\begin{equation}
\sup_{\theta\in\Theta}\ \sup_{t\in\mathbb{R}}\,
\left|
\Pr\!\left(\frac{R_{M,\theta}-\mu^\circ(\theta)}{\sqrt{v^\circ(\theta)/M}}\le t\ \middle|\ \theta\right)-\Phi(t)
\right|
\ \le\ \frac{C'}{\sqrt{M}},
\end{equation}
for a constant $C'$ depending only on the uniform moment bound.
\end{lemma}

\begin{lemma}
\label{lem:proxy}
Under Assumption~\ref{ass:uniform-moments-fixed}, define the effective variance
\[
\tau^2_{\mathrm{eff}}(\theta;M,\tau)\;=\;\tau^2+\frac{v^\circ(\theta)}{M}.
\]
Then, uniformly in $\theta\in\Theta$,
\begin{equation}
\label{eq:proxy-main}
L_M(\theta;\tau)
\ =\ \sqrt{\frac{\tau^2}{\tau^2_{\mathrm{eff}}(\theta;M,\tau)}}\,
\exp\!\left\{-\frac{\mu^\circ(\theta)^2}{2\,\tau^2_{\mathrm{eff}}(\theta;M,\tau)}\right\}
\,+\, r_{M,\tau}(\theta),
\end{equation}
with $\sup_{\theta\in\Theta}|r_{M,\tau}(\theta)|\le C''/\sqrt{M}$ for a constant $C''<\infty$ that depends only on the uniform moment bound in Assumption~\ref{ass:uniform-moments-fixed}.
\end{lemma}

Proof sketch. We write $R_{M,\theta}=\mu^\circ(\theta)+\sqrt{v^\circ(\theta)/M}\,Z$ with $Z$ standardized. Evaluate $\mathbb{E}[\exp\{-R_\theta^2/(2\tau^2)\}]$ by completing the square under a Gaussian proxy for $Z$ and control the approximation error uniformly via Lemma~\ref{lem:be}; continuity and compactness make the constants uniform.

\begin{lemma}[Population Laplace bound]
\label{lem:laplace}
Let $U\supset\Theta^\dagger$ be open and let $\eta=\eta(U)>0$ be the separation
constant from Assumption~\ref{ass:ident-fixed}. Under Assumptions~\ref{ass:uniform-moments-fixed} and~\ref{ass:exact_repres}, there exist constants $C(U)\in(0,\infty)$ and $C'''\in(0,\infty)$, independent of
$M,\tau$, and constants $M_0\in\mathbb{N}$ and $K>0$ such that for all
$M\ge M_0$ and all $\tau>0$ with $M\tau^2\ge K$,
\[
\frac{\int_{\Theta\setminus U} p(\theta)\,L_M(\theta;\tau)\,{\rm d}\theta}
     {\int_{U} p(\theta)\,L_M(\theta;\tau)\,{\rm d}\theta}
\ \le\
C(U)\,\exp\!\left\{-\frac{\eta}{4\,\bigl(\tau^2+v_{\max}/M\bigr)}\right\}
\left(1+\frac{C'''}{\sqrt{M}}\right),
\]
\end{lemma}

\paragraph{Remarks.} The mild coupling $M\tau^2\ge K$ prevents the prefactor
$\sqrt{\tau^2/(\tau^2+v_{\max}/M)}$ in the denominator from collapsing as
$\tau\downarrow 0$ at fixed $M$. It is standard in kernel-Laplace analyses and is compatible with the concentration regime $\tau\downarrow 0$, $M\to\infty$ with $M\tau^2\to\infty$.

Finally, we prove the main concentration result.

\begin{theorem}[Concentration at fixed $\beta^\circ$]
\label{thm:concentration-fixed}
Assume \ref{ass:uniform-moments-fixed}–\ref{ass:ident-fixed}. Let $M\to\infty$ and $\tau=\tau_M\downarrow 0$ such that
\begin{equation} \label{eq:conc-rate}
M\tau^2 \longrightarrow\ \infty.
\end{equation}
Then for any open $U\supset\Theta^\dagger$ and any closed $F$ with $F\cap\Theta^\dagger=\emptyset$,
\begin{equation}
\pi_{M,\tau_M}(U)\ \longrightarrow\ 1,\qquad \pi_{M,\tau_M}(F)\ \longrightarrow\ 0.
\end{equation}
Moreover, if in addition $n_\theta\to\infty$, then almost surely
\begin{equation}
\Pi_{n_\theta}(U)\ \longrightarrow\ 1,\qquad \Pi_{n_\theta}(F)\ \longrightarrow\ 0.
\end{equation}
\end{theorem}

\paragraph{Remarks.}
(i) Taking $M\to\infty$ alone is not sufficient: if $\tau_M\equiv \tau_0>0$ then $\pi_{M,\tau_0}$ is diffuse and does not concentrate. (ii) A practical sufficient scaling is $M\tau_M^2\to\infty$, which implies $\tau_M^2+v_{\max}/M\to0$ and keeps the Gaussian proxy uniform. (iii) The scaling $M\tau_M^2\to\infty$ guarantees the denominator lower bound in Lemma~\ref{lem:laplace} does not degenerate as $\tau_M\downarrow 0$ at finite $M$, and is compatible with the desired regime $M\to\infty$, $\tau_M\downarrow 0$.

Recall that $\beta^\circ$ denotes the best linear projection under the \textit{observed} law $\mathcal L_{\mathrm{obs}}(X,Y)$ and
\begin{equation}
\mu^\circ(\theta)\ :=\ \mathbb{E}_{\mathrm{sim}}\!\big[Y-\beta^{\circ\top}\mathbf X \,\big|\, \theta\big],
\qquad
\Theta^\dagger\ :=\ \arg\min_{\theta\in\Theta}\ \mu^\circ(\theta)^2.
\end{equation}
Let
\begin{equation}
\Theta^\star\ :=\ \{\theta\in\Theta:\ \mathcal L_{\mathrm{sim}}(X,Y\mid\theta)=\mathcal L_{\mathrm{obs}}(X,Y)\}
\end{equation}
be the (possibly set-valued) collection of parameters whose simulator reproduces the observed joint law. This implies that the simulator is correct specification of data, however, the linear projection can be still misspecified.  

\begin{proposition} \label{prop:star-in-dagger}
If $\theta\in\Theta^\star$ implies $\mathcal L_{\mathrm{sim}}(X,Y\mid\theta)=\mathcal L_{\mathrm{obs}}(X,Y)$, then
\begin{equation}
\Theta^\star\ \subseteq\ \Theta^\dagger.
\end{equation}
\end{proposition}

The set $\Theta^\dagger$ can be \textit{strictly larger} than $\Theta^\star$ whenever the residual moment is not separating, i.e., there exist $\theta\notin\Theta^\star$ with $\mu^\circ(\theta)=0$ (under-identification by the linear summary). In such cases the pseudo-posterior concentrates on the larger pseudo-true set $\Theta^\dagger$, potentially a low-dimensional manifold. This result shows the key limitation of the proposed loss function as in most practical applications point identification would be infeasible. We will demonstrate this property in the following experiment.

\section{Numerical Simulation-Based Inference Experiment}
\label{sec:toy_experiment}

We consider a controlled two-dimensional SBI problem in which the data-generating process is nonlinear in both the latent parameters and the observed covariates, while the surrogate model used to construct the pseudo-posterior is intentionally misspecified. The purpose of this experiment is to compare the proposed batched residual weighting method to the exact Bayesian posterior in a setting where the likelihood is known and can be evaluated numerically.

\paragraph{Generative Model.} Let $\theta = (\theta_0, \theta_1) \in \mathbb{R}^2$
denote the unknown parameter vector. We place an isotropic Gaussian prior
\begin{equation} \label{eq:prior_experiment_1}
\pi(\theta) = \mathcal{N}(0, 5^2 I_2).    
\end{equation}
Given $\theta$, observations are generated according to the following hierarchical model. First, latent covariates $X_i$ are drawn from a log-normal distribution,
\begin{equation}
\log X_i \mid \theta \sim \mathcal{N}\!\left(\frac{\theta_1}{5},\, 0.5^2\right),
\label{eq:lognormal_x}
\end{equation}
which guarantees $X_i > 0$ almost surely. Conditional on $X_i$ and $\theta$, the response variable $Y_i$ has conditional mean
\begin{equation}
\bar{Y}_i(\theta, X_i) = \theta_1 \log X_i + \theta_0 X_i,
\label{eq:conditional_mean}
\end{equation}
and is observed with additive Gaussian noise,
\begin{equation}
Y_i \mid X_i, \theta \sim \mathcal{N}\!\left(\bar{Y}_i(\theta, X_i),\, 1\right),
\label{eq:y_given_x}
\end{equation}
for $i = 1, \dots, n$. Observed data are simulated at the true parameter value $\theta_{\mathrm{true}} = (2, 2)$, with sample size $n_{\mathrm{obs}} = 200$. Figure~\ref{fig:raw_data} shows the resulting scatter plot of $(X, Y)$ with a clear curvature inconsistent with a linear conditional mean.

\begin{figure}
    \centering
    \label{fig:raw_data}
    \includegraphics[width=0.65\linewidth]{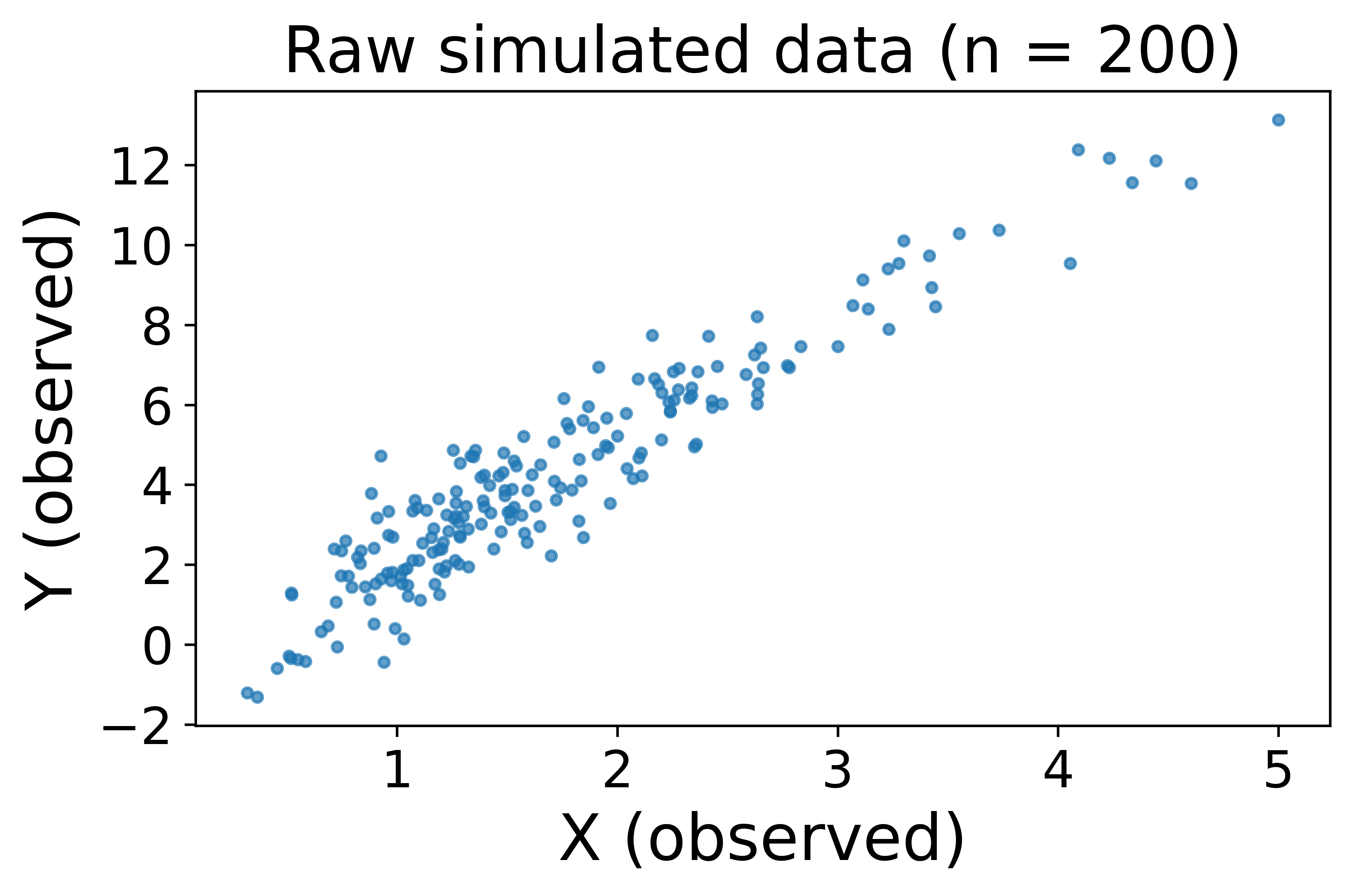}
    \vspace{-6mm}
    \caption{A realization of simulated observed data.}
\end{figure}
\paragraph{Exact Posterior via Metropolis--Hastings.} Because the likelihood implied by Equations~\eqref{eq:lognormal_x}--\eqref{eq:y_given_x} is analytically tractable, we compute a reference posterior distribution using a random-walk Metropolis--Hastings (RWMH) sampler and use it to benchmark our pseudo-posterior. For observed data $\{(x_i, y_i)\}_{i=1}^n$, the log-likelihood is
\begin{equation}
\log p(y \mid x, \theta) =
    \sum_{i=1}^n \log \mathcal{N}\!\left( y_i \,\middle|\, \theta_1 \log x_i + \theta_0 x_i,\, 1 \right),
\label{eq:log_likelihood}
\end{equation}
which is combined with the Gaussian prior of Equation~\eqref{eq:prior_experiment_1} to form the posterior. The RWMH chain is initialized at the prior mean, run for $4\times 10^4$ iterations, and the first $5\times 10^3$ samples are discarded as burn-in. The remaining samples are treated as draws from the exact posterior and used for quantitative and qualitative comparison with the SBI pseudo-posterior.

\paragraph{Surrogate Regression Model.} To define a low-dimensional discrepancy, we introduce a linear surrogate model that is intentionally misspecified relative to the true generative process. Specifically, we fit an OLS regression of $Y$ on an intercept and $X$,
\begin{equation}
Y \approx \beta_0 + \beta_1 X,
\label{eq:aux_model}
\end{equation}
using the observed data. The estimated coefficients $\hat{\beta} = (\hat{\beta}_0, \hat{\beta}_1)$ are computed once and held fixed throughout the inference procedure. This surrogate model ignores the logarithmic dependence on $X$ in Equation~\eqref{eq:conditional_mean} and therefore cannot represent the true conditional mean exactly.

\begin{figure}[ht!]
    \centering
    \includegraphics[width=0.98\linewidth]{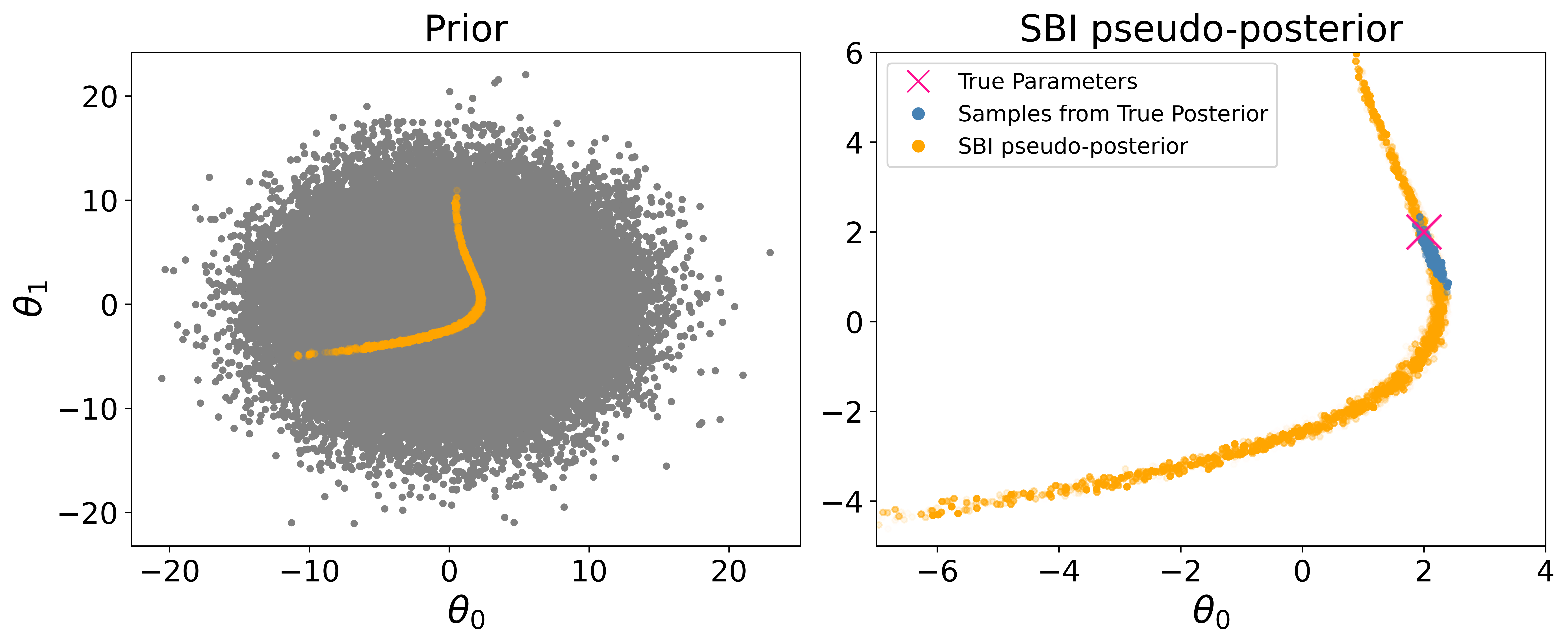}
    \caption{{\bf Left Panel:} Prior samples in $(\theta_0, \theta_1)$-space (gray), together with the support of the SBI pseudo-posterior induced by batched residual weighting (orange). The prior is diffuse and isotropic, while the pseudo-posterior concentrates along a curved one-dimensional manifold (identified set).  {\bf Panel Right:} Comparison between the exact posterior obtained via Metropolis--Hastings (blue), the SBI pseudo-posterior (orange), and the true parameter value $\theta_{\mathrm{true}}$ (pink cross). The pseudo-posterior exhibits strong concentration near the true posterior mass in one dimension while in the other dimension the identification is not possible.}
    \label{fig:prior_vs_sbi}
\end{figure}

\paragraph{Evaluation.} We evaluate the performance of the proposed method by comparing the SBI pseudo-posterior to the exact posterior. Figure~\ref{fig:prior_vs_sbi} visualizes the effect of the batched residual kernel weighting on the parameter space. The left panel shows samples drawn from the Gaussian prior together with the subset of parameters receiving non-negligible weight under the SBI pseudo-posterior. While the prior is diffuse and nearly isotropic, the pseudo-posterior mass collapses onto a curved one-dimensional manifold in $(\theta_0, \theta_1)$-space. This geometry reflects the set of parameters that induce simulated data whose surrogate regression residuals match those observed in the real data.

The right panel compares the SBI pseudo-posterior to the exact Bayesian posterior obtained via Metropolis--Hastings sampling. There exists a strong correlation between the true posterior of $\theta_0$ and $\theta_1$, arising from the nonlinear dependence of the conditional mean $\theta_1 \log X + \theta_0 X$. The pseudo-posterior covers the true posterior despite being constructed using a misspecified linear surrogate model. However, the pseudo-posterior concentrate to a narrow one-dimensional manifold, illustrating that the batched residual criterion identifies parameters that reproduce the observed surrogate behavior rather than the full likelihood. This experiment demonstrate that the pseudo-posterior need not approximate the full posterior density everywhere, but instead concentrates on a lower-dimensional set of parameters that are observationally indistinguishable under the chosen surrogate statistic. In this sense, the geometry of the pseudo-posterior provides direct insight into identifiability and degeneracy induced by the surrogate model. This is the key limitation of the proposed method. Hence, due to its computational tractability it is appropriate for exploratory analysis, but if the goal is to estimate the actual posterior this method can lead to idetifiability issues.

\section{Illustrative Example: Cosmology Inference with Galaxy Properties}
\label{sec:example}

Forward models of galaxy formation encode complex, nonlinear couplings between astrophysical processes and observable galaxy properties, making direct likelihood-based inference over their high-dimensional parameter spaces intractable \citep{tortorelli2025galsbi}. Simulation suites such as the DaRk mattEr and Astrophysics with Machine learning Simulations Project \citep[DREAMS;][]{rose2025introducing,rose2025dreams_sims} offer a controlled setting in which uncertain physical parameters can be systematically varied, enabling statistical calibration against observed galaxy scaling relations. The DREAMS suite consists of 1,024 hydrodynamic simulations with the IllustrisTNG \citep[TNG;][]{2018Pillepicha, 2018Weinberger} galaxy formation model that models many physical processes such as star formation, stellar feedback, active galactic nuclei (AGN) feedback, supernovae (SN) feedback, chemical enrichment, among others. In this section, we utilize the Cold Dark Matter suite of Milky-Way-Mass galaxy simulations from DREAMS to identify regions of parameter space that yield physically plausible galaxy populations and quantify which aspects of subgrid physics are most strongly constrained by current data. We perform a pseudo-posterior analysis described here to identify regions of parameter space that best reproduce observed galaxy scaling relations, including the stellar mass-halo mass (SMHM) relation and the Black Hole Mass-Stellar Velocity Dispersion ($M$--$\sigma$) relation. The SMHM relation is measured across time, and we use two measurements to constrain these parameters, one at the present day ($z=0$) and one $\sim7.5$~Gyr ago ($z=1$). The data for the SMHM relation are taken from \cite{2019Behroozi}, and the data from the $M$--$\sigma$ relation are taken from \cite{2020Greene}.

\paragraph{Data.} The TNG model, used to create DREAMS, includes $\mathcal{O}$(100) physical and numerical parameters that impact galaxy formation. The DREAMS simulations specifically vary three uncertain astrophysical and two cosmological parameters that impact galaxy formation. Specifically, these vary the matter density $\Omega_m~\in~[0.274,0.354]$, the amplitude of large scale matter fluctuations $\sigma_8~\in~[0.780,0.888]$, the specific energy of SN feedback $\bar{e}_w~\in~[0.9,14.4]$, the speed of SN winds $\kappa_w~\in~[3.7,14.8]$, and the AGN feedback coupling to surrounding gas $\epsilon_{f,\mathrm{high}}~\in~[0.025,0.4]$. Hence, $\theta = \{\Omega_m, \sigma_8, \bar{e}_w, \kappa_w,\epsilon_{f,\mathrm{high}}\}$. The effect of cosmological parameters, $\Omega_m$ and $\sigma_8$, on present-day galactic properties are negligible, so we focus our analysis on constraining the three astrophysical parameters.

We use the conditional normalizing flow simulator described in \cite{nguyen2024dreams} as our data generator to sample galaxy properties considered here \citep[also see,][]{rose2025dreams_sims}. We follow the prescription described here, and compute the non-normalized weights by a Gaussian likelihood function
\begin{equation}
w_i = \mathrm{exp} \left( - \frac{(\sum_{i=1}^{M} y_{i,k} - \mu_{\mathrm{obs},i,k})^2}{2 M^2 \tau_k^2} \right)
\end{equation}
where $y_{i,k}$ is the predicted property for galaxy $i$ in relation $k$, $\mu_{\mathrm{obs},k}$ is the mean observational value, and $\tau_k$ is a free temperature parameter set to the uncertainty derived from the observed relation.

\begin{figure}[ht!]
    \centering
    \includegraphics[width=0.98\linewidth]{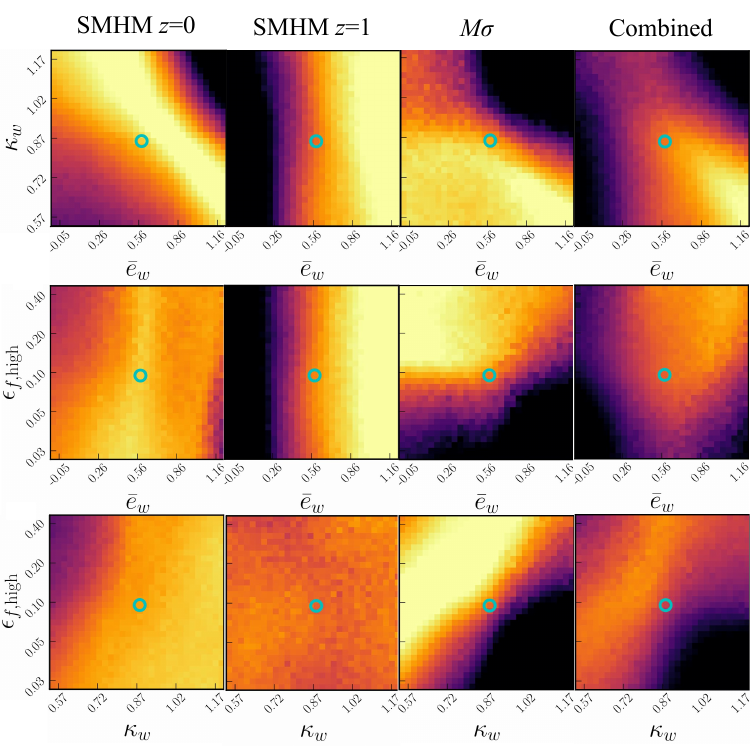}
    \caption{Marginalized 2D pseudo-posterior distributions for the three uncertain astrophysical parameters: the specific energy of supernova feedback ($\bar{e}_w$), the speed of supernova winds ($\kappa_w$), and the AGN feedback coupling strength ($\epsilon_{f,\mathrm{high}}$).
    From left to right, the columns display constraints derived from the SMHM relation at $z=0$, the SMHM relation at $z=1$, the black hole mass--stellar velocity dispersion ($M$--$\sigma$) relation, and the combined constraints.
    Brighter colors indicate regions of higher weight, and the cyan annuli indicate the fiducial parameters used in this model.
    The degeneracy between $\bar{e}_w$ and $\kappa_w$ visible in the $z=0$ SMHM relation is broken by the inclusion of the $z=1$ data.}
    \label{fig:constraints}
\end{figure}

\paragraph{Results.} Figure~\ref{fig:constraints} presents the marginalized pseudo-posteriors derived from these weights. Each row shows the marginalized pseudo-posteriors for two of the three astrophysical parameters, $\bar{e}_w$, $\kappa_w$, and $\epsilon_{f,\mathrm{high}}$. Each column shows the marginalized constraints from a different scaling relation, and the right column shows the combined weights from all individual scaling relations.

At the present time, $z=0$, SMHM relation and the $M$--$\sigma$ relation, the two SN feedback parameters, $\bar{e}_w$ and $\kappa_w$, are degenerate across a significant section of parameter space. This degeneracy is expected as both parameters regulate the ejection of gas from star-forming regions. However, the effect of the wind speed parameter, $\kappa_w$, is diminished at earlier times due to a minimum wind speed floor that these galaxies only exceed near the present day. This speed floor results in little dependence of the galactic properties at early times making the wind energy parameter, $\bar{e}_w$, dominate changes to the galactic properties, which can be seen in the SMHM relation at $z=1$. Adding in these constraints from earlier in the universe breaks the degeneracy between $\bar{e}_w$ and $\kappa_w$, showing a clear preference for low values of $\kappa_w$ and high values of $\bar{e}_w$ in the top right panel of Figure~\ref{fig:constraints}.

The AGN parameter, $\epsilon_{f,\mathrm{high}}$, is not strongly constrained by the SMHM relation at any time due to the low mass of these galaxies. The $M$--$\sigma$ relation, on the other hand, directly compares the mass of the BH producing the AGN feedback to observations and provides much tighter constraints on $\epsilon_{f,\mathrm{high}}$, preferring higher values. This results in a preference for high values of $\epsilon_{f,\mathrm{high}}$ and $\bar{e}_w$ in the combined pseudo-posteriors, but still contains a degeneracy between $\epsilon_{f,\mathrm{high}}$ and $\kappa_w$ that would require additional scaling relations to break. These results shows that the pseudo-posteriors here we can draw meaningful physical insight with a fraction of the effort needed to model actual data. 

\begin{figure}[ht!]
    \centering
    \includegraphics[width=\linewidth]{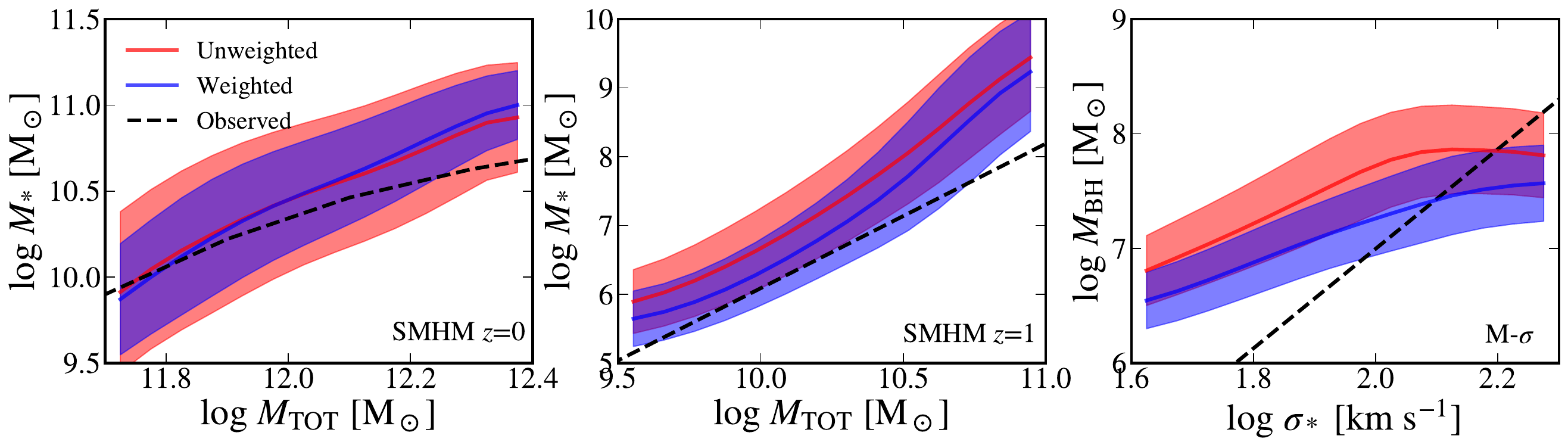}
    \caption{Comparison of simulated galaxy scaling relations with observational data before and after applying pseudo-posterior weights.
    The panels display the $z=0$ SMHM relation (left), the $z=1$ SMHM relation (middle), and the $M$--$\sigma$ relation (right).
    The red shaded regions represent the unweighted distribution of the simulations and the blue shaded regions show the distributions after applying the combined weights shown in the rightmost panel of Figure~\ref{fig:constraints}.
    The black dashed lines indicate the mean observed relations taken from \cite{2019Behroozi} and \cite{2020Greene}.
    The weighted distributions, particularly the $z=1$ SMHM and $M$--$\sigma$ relations, shift significantly toward the observations, demonstrating that the pseudo-posterior method correctly identifies regions of parameter space that better match the observed relations.}
    \label{fig:relations}
\end{figure}

To validate the pseudo-posterior distributions, we compare the weighted and unweighted galaxy populations against the observational data. Figure~\ref{fig:relations} shows the $z=0$ SMHM relation in the left panel, the $z=1$ SMHM relation in the middle panel, and the $M-\sigma$ relation in the right panel. Weighted relations use the combined weights from the right column of Figure~\ref{fig:constraints}. The unweighted distributions (red) are displaced from the observed mean relation for all three shown in Figure~\ref{fig:relations}, indicating the need to place tighter constraints on the simulation parameters than the uniform priors used to create this dataset.

Upon applying the weights, the weighted $z=1$ SMHM relation and $M$--$\sigma$ distributions (blue) shift toward the observed mean relations (black). The $z=0$ SMHM relation does not show a significant shift once the weights are applied, but still shows a better agreement with the observed relation. 
This agreement suggests that the high-weight regions of the parameter space correspond to physically more plausible realizations of the simulation parameter space. However, the degeneracy that remains between the simulation parameters and the offsets still present between the weighted and observed relations indicates that both additional scaling relations could be added to further constrain these parameters, and that additional simulation parameters could be constrained to better match the observations.

\section{Discussion and Conclusion}
\label{sec:conclusion}

We have presented a theoretical analysis of a surrogate-based method for simulation-based inference and demonstrated its application in astrophysics data analysis. Our analysis established a central limit theorem, consistency of the empirical pseudo-posterior, stability with respect to estimation of the surrogate regression coefficient, and concentration on simulator parameters that reproduce the observed linear relationship. Our theoretical results characterize the conditions under which the proposed method works and under which it may fail. 
Beyond the specific application, the framework highlights how simple, interpretable data projection summary statistics, in this case, the best linear projection under the observed law, can provide a computationally efficient and privacy-preserving interface between observational data and large-scale scientific simulators. By requiring access only to the surrogate coefficient $\widehat{\beta}$, the method enables inference even when the underlying observational dataset cannot be shared directly -- whether due to privacy or governance constraints, computational intractability, or the difficulty of accurately modeling observational effects such as measurement error, sample selection, and data dependencies.

At the same time, our analysis highlights the key limitations of the proposed loss function and clarifies its proper place within the broader landscape of likelihood-free inference. Because the pseudo-posterior is constructed from a single summary statistic, full point identification of the simulator parameters is generally not guaranteed. We documented scenarios in which distinct parameter values produce identical expected residual means, leading the pseudo-posterior to concentrate on a low-dimensional manifold that may be strictly larger than the true parameter set. This reflects a fundamental property of moment-based or summary-based inference: identifiability is limited by the informativeness of the chosen summary and loss function, and additional structure is required when sharper inference is needed. For this reason, the sudo-posterior is best viewed as a fast, stable, and transparent calibration mechanism that aligns a simulator with a chosen data projection. It serves as a practical first layer in a multi-stage inference pipeline, complementing richer summary statistics, additional surrogate models, or more expressive simulation-based inference techniques when scientific goals require greater fidelity or full Bayesian posterior precision.

While pseudo-posteriors provide a computationally efficient and privacy-preserving alternative to full likelihood-based inference and can be used for efficient exploratory analysis, model checking, and preliminary calibration, they should not be interpreted as providing the same level of inferential resolution as an exact Bayesian posterior. Instead, they offer a principled starting point that can be integrated with more informative summaries or advanced SBI methods to achieve greater statistical power and tighter identification in complex generative modeling tasks.

\section*{Acknowledgments}
 We thank Peter M\"uller, Pratik Patil, and Abhra Sarkar for helpful discussions. AF and PT acknowledge support from the National Science Foundation under Cooperative Agreement 2421782 and the Simons Foundation grant MPS-AI-00010515 awarded to NSF-Simons AI Institute for Cosmic Origins (CosmicAI, \href{https://www.cosmicai.org/}{https://www.cosmicai.org/}).

 \appendix 
 \section{Proofs}

\begin{proof}[Proof of Lemma~\ref{lem:clt}]
Fix $\theta\in\Theta$ and define 
\[
Z_m = Y_m^{\mathrm{sim}} - \widehat{\beta}^{\top}\mathbf{X}_m^{\mathrm{sim}},
\qquad m=1,\dots,M.
\]
By Assumption~\ref{ass:moments}(i), $(Z_m)_{m=1}^M$ are i.i.d.\ given~$\theta$ with
\[
\mathbb{E}[Z_m\mid\theta]=\mu(\theta), 
\qquad 
\Var(Z_m\mid\theta)=v(\theta) < \infty.
\]
Define $R_{M,\theta}=M^{-1}\sum_{m=1}^M Z_m$ and note that
\[
\sqrt{M}\,\frac{R_{M,\theta}-\mu(\theta)}{\sqrt{v(\theta)}}
=\frac{1}{\sqrt{M}}\sum_{m=1}^M \frac{Z_m-\mu(\theta)}{\sqrt{v(\theta)}}.
\]

By Assumption~\ref{ass:moments}(ii), 
$\mathbb{E}[|Z_m-\mu(\theta)|^4\mid\theta]<\infty$. 
The Lyapunov fraction with exponent $\delta=2$ is
\[
\frac{\sum_{m=1}^M \mathbb{E}[|Z_m-\mu(\theta)|^{2+\delta}\mid\theta]}
     {\left(\sum_{m=1}^M \Var(Z_m\mid\theta)\right)^{1+\delta/2}}
=
\frac{M\,\mathbb{E}[|Z-\mu(\theta)|^{4}\mid\theta]}
     {\bigl(M\,v(\theta)\bigr)^{2}}
=
\frac{\mathbb{E}[|Z-\mu(\theta)|^{4}\mid\theta]}{M\,v(\theta)^{2}}
\xrightarrow[M\to\infty]{}0.
\]
Hence Lyapunov’s condition holds and the Lindeberg--Feller central limit theorem applies:
\[
\sqrt{M}\,\frac{R_{M,\theta}-\mu(\theta)}{\sqrt{v(\theta)}}
\ \xRightarrow[M\to\infty]{}\ \mathcal{N}(0,1)
\quad\text{(conditional on $\theta$)}.
\]

By linearity of transformations,
\[
R_{M,\theta}\ \xRightarrow[M\to\infty]{}\ 
\mathcal{N}\!\bigl(\mu(\theta),\,v(\theta)/M\bigr)
=\mathcal{N}\!\bigl(\mu(\theta),\,s^2(\theta)\bigr),
\]
which proves the first claim. Let $\varphi$ be bounded and continuous. 
Convergence in distribution and the Portmanteau theorem imply
\[
\mathbb{E}\!\left[\varphi(R_{M,\theta})\mid\theta\right]
\longrightarrow
\int \varphi(t)\,
\mathcal{N}\!\bigl(t\mid \mu(\theta), s^2(\theta)\bigr)\,\mathrm{d}t,
\]
establishing the second claim. This completes the proof.
\end{proof}

\begin{proof}[Proof of Lemma~\ref{lem:L_M-gauss}.]
Suppose $Z:=R_{M,\theta}\sim\mathcal{N}(\mu,s^2)$ with $\mu=\mu(\theta)$ and $s^2=s^2(\theta)$.
Then $Z^2/s^2$ follows a noncentral chi-square distribution with one degree of freedom and
noncentrality parameter $\lambda=(\mu/s)^2$. Hence, for any $a>-1/2$,
\[
\mathbb{E}\!\left[e^{-a\,Z^2}\right]
=
\frac{1}{\sqrt{1+2a s^2}}\,
\exp\!\left\{-\frac{\mu^2\,a}{1+2a s^2}\right\}.
\]
Setting $a=(2\tau^2)^{-1}$ gives
\[
\mathbb{E}\!\left[\exp\!\left\{-\frac{Z^2}{2\tau^2}\right\}\right]
=
\frac{1}{\sqrt{1+s^2/\tau^2}}\,
\exp\!\left\{-\frac{\mu^2}{2(\tau^2+s^2)}\right\}
=
\sqrt{\frac{\tau^2}{\tau^2+s^2}}\,
\exp\!\left\{-\frac{\mu^2}{2(\tau^2+s^2)}\right\},
\]
which yields the stated expression.
\end{proof}

\begin{proof}[Proof of Lemma~\ref{lem:uniform-conv}]
By Lemma~\ref{lem:L_M-gauss}, for each $\theta$,
\[
L_M(\theta;\tau)
= \phi\bigl(s^2(\theta),\mu(\theta)\bigr),
\qquad
L_\infty(\theta;\tau) = \phi\bigl(0,\mu(\theta)\bigr),
\]
where, with $t\ge 0$ and $m\in\mathbb{R}$,
\[
\phi(t,m)\ :=\ \sqrt{\frac{\tau^2}{\tau^2+t}}\,
\exp\!\left\{-\frac{m^2}{2(\tau^2+t)}\right\}.
\]
Note $s^2(\theta)=v(\theta)/M\le V/M$ for all $\theta$.

Fix $\theta$ and apply the mean value theorem in the variable $t$.
For some $\xi\in(0,s^2(\theta))$,
\[
\bigl|L_M(\theta;\tau)-L_\infty(\theta;\tau)\bigr|
= \bigl|\phi\bigl(s^2(\theta),\mu(\theta)\bigr)-\phi\bigl(0,\mu(\theta)\bigr)\bigr|
= \bigl|\partial_t \phi(\xi,\mu(\theta))\bigr|\; s^2(\theta).
\]
We now bound $\sup_{m\in\mathbb{R}}|\partial_t\phi(t,m)|$ uniformly in $t\in[0,V/M]$.
A direct calculation gives, with $c=\tau^2+t$,
\[
\partial_t \phi(t,m)
= \frac{\tau}{2\,c^{3/2}}\,
\exp\!\left\{-\frac{m^2}{2c}\right\}\,
\Bigl(-1+\frac{m^2}{c}\Bigr).
\]
Let $u=m^2/c\ge 0$. Then
\[
\bigl|\partial_t \phi(t,m)\bigr|
\le \frac{\tau}{2\,c^{3/2}}\,
\exp\!\left(-\frac{u}{2}\right)\,\bigl(1+u\bigr).
\]
The function $g(u)=(1+u)e^{-u/2}$ on $u\ge 0$ attains its maximum at $u=1$ with
$g(1)=2e^{-1/2}$. Hence
\[
\sup_{m\in\mathbb{R}}\bigl|\partial_t \phi(t,m)\bigr|
\le \frac{\tau}{2\,c^{3/2}}\cdot 2e^{-1/2}
= \frac{\tau e^{-1/2}}{(\tau^2+t)^{3/2}}
\le \frac{\tau e^{-1/2}}{\tau^3}
= \frac{e^{-1/2}}{\tau^{2}},
\]
since $t\ge 0$. Therefore,
\[
\bigl|L_M(\theta;\tau)-L_\infty(\theta;\tau)\bigr|
\le \frac{e^{-1/2}}{\tau^{2}}\; s^2(\theta)
\le \frac{e^{-1/2}}{\tau^{2}}\;\frac{V}{M}.
\]
Taking the supremum over $\theta\in\Theta$ yields
\[
\sup_{\theta\in\Theta}\,\bigl|L_M(\theta;\tau)-L_\infty(\theta;\tau)\bigr|
\le \frac{e^{-1/2}V}{\tau^{2} M}\ \xrightarrow[M\to\infty]{}\ 0,
\]
which proves both the claim and the uniform $\mathcal{O}(1/M)$ rate.
\end{proof}

\begin{proof}[Proof of Theorem~\ref{thm:sn}]
Let $U_j:=\widetilde w_j\,h(\theta_j)$ and $V_j:=\widetilde w_j$. 
Conditional on $\theta_j$, the batch average $R_{\theta_j}$ is generated independently across $j$, so 
$\{(U_j,V_j)\}_{j\ge1}$ are i.i.d. (with respect to the product of the prior $p$ and the simulator randomness).
By the strong law of large numbers and the assumed integrability,
\[
\frac{1}{n_\theta}\sum_{j=1}^{n_\theta} U_j \xrightarrow{\text{a.s.}} \mathbb{E}[U_1],
\qquad
\frac{1}{n_\theta}\sum_{j=1}^{n_\theta} V_j \xrightarrow{\text{a.s.}} \mathbb{E}[V_1],
\]
with $\mathbb{E}[V_1]=\mathbb{E}[L_M(\Theta;\tau)]>0$. 
Hence, by the continuous mapping theorem,
\[
\sum_{j=1}^{n_\theta} w_j\,h(\theta_j)
=\frac{\sum_{j=1}^{n_\theta} U_j}{\sum_{j=1}^{n_\theta} V_j}
\xrightarrow{\text{a.s.}}
\frac{\mathbb{E}[U_1]}{\mathbb{E}[V_1]}.
\]

It remains to identify the limits. Using iterated expectations,
\[
\mathbb{E}[U_1]
=\mathbb{E}\!\left[\mathbb{E}\!\left[\widetilde w_1\,h(\theta_1)\mid \theta_1\right]\right]
=\mathbb{E}\!\left[h(\Theta)\,\mathbb{E}\!\left[e^{-R_{M,\Theta}^2/(2\tau^2)}\mid \Theta\right]\right]
=\mathbb{E}\!\left[h(\Theta)\,L_M(\Theta;\tau)\right],
\]
and similarly $\mathbb{E}[V_1]=\mathbb{E}[L_M(\Theta;\tau)]$. Therefore,
\[
\frac{\mathbb{E}[U_1]}{\mathbb{E}[V_1]}
=\frac{\int h(\theta)\,L_M(\theta;\tau)\,p(\theta)\,\mathrm{d}\theta}
       {\int L_M(\theta;\tau)\,p(\theta)\,\mathrm{d}\theta}
=\int h(\theta)\,\pi_{M,\tau}(\mathrm{d}\theta),
\]
which proves the claim.
\end{proof}

\begin{proof}[Proof of Theorem~\ref{thm:PhiM-to-PhiInf}.]
Write $Z_M:=\int L_M(\theta;\tau)\,p(\theta)\,{\rm d}\theta$ and
$Z_\infty:=\int L_\infty(\theta;\tau)\,p(\theta)\,{\rm d}\theta$.
Since $0\le L_M,L_\infty\le 1$, both $Z_M,Z_\infty\in(0,1]$ and are finite.
By Lemma~\ref{lem:uniform-conv}, $\|L_M-L_\infty\|_\infty\to0$, whence
$|Z_M-Z_\infty|\le\|L_M-L_\infty\|_\infty\to0$ and thus $Z_M\to Z_\infty>0$.

Assume $|h|\le K<\infty$. Then
\begin{align*}
\Phi_M(h)-\Phi_\infty(h)
&=\frac{\int h L_M\,{\rm d}p}{Z_M}-\frac{\int h L_\infty\,{\rm d}p}{Z_\infty}\\
&=\frac{\int h(L_M-L_\infty)\,{\rm d}p}{Z_M}
+\Bigl(\frac{1}{Z_M}-\frac{1}{Z_\infty}\Bigr)\int h L_\infty\,{\rm d}p,
\end{align*}
so that
\[
|\Phi_M(h)-\Phi_\infty(h)|
\le \frac{\int |h|\,|L_M-L_\infty|\,{\rm d}p}{Z_M}
+\frac{|Z_M-Z_\infty|}{Z_M Z_\infty}\,\int |h|\,L_\infty\,{\rm d}p.
\]
Using $\int |h|\,|L_M-L_\infty|\,{\rm d}p\le K\,\|L_M-L_\infty\|_\infty$ and
$\int |h|\,L_\infty\,{\rm d}p\le K Z_\infty$, we get
\[
|\Phi_M(h)-\Phi_\infty(h)|
\le \frac{K}{Z_M}\,\|L_M-L_\infty\|_\infty
+\frac{K}{Z_M}\,|Z_M-Z_\infty|
\le \frac{2K}{Z_M}\,\|L_M-L_\infty\|_\infty.
\]
Hence $|\Phi_M(h)-\Phi_\infty(h)|\to0$ as $M\to\infty$; if, moreover,
$\|L_M-L_\infty\|_\infty=\mathcal{O}(M^{-1})$, the same display yields the
$\mathcal{O}(M^{-1})$ rate for bounded $h$.
\end{proof}

\begin{proof}[Proof of Theorem~\ref{thm:aux}.]
Define the random map
\[
T(\beta)\ :=\ \frac{N(\beta)}{Z(\beta)},
\qquad
N(\beta):=\int h(\theta)\,L_M^{(\beta)}(\theta;\tau)\,p(\theta)\,\mathrm d\theta,\quad
Z(\beta):=\int L_M^{(\beta)}(\theta;\tau)\,p(\theta)\,\mathrm d\theta.
\]
The claim is $T(\widehat{\beta})\xrightarrow{P}T(\beta^\circ)$.

For any realization $\{(X_m,Y_m)\}_{m=1}^M$, the map $\beta\mapsto \exp\!\big\{-\tfrac{1}{2\tau^2}\big(M^{-1}\sum_{m=1}^M (Y_m-\beta^\top\mathbf X_m)\big)^2\big\}$ is continuous and bounded by $1$. Since the simulator law of $(X_m,Y_m)$ does not depend on $\beta$, by dominated convergence $L_M^{(\beta)}(\theta;\tau)$ is continuous in $\beta$ for each fixed $\theta$.

By Assumption~\ref{ass:stability}(ii) there exists an integrable envelope $g$ with $L_M^{(\beta)}(\theta;\tau)\le g(\theta)$ for all $\beta$ in a neighborhood $\mathcal N$ of $\beta^\circ$. With $h$ and $L_M^{(\beta)}(\theta;\tau)$ bounded, dominated convergence in $\theta$ yields $N(\beta)\to N(\beta^\circ)$ and $Z(\beta)\to Z(\beta^\circ)$ as $\beta\to\beta^\circ$. Moreover, $0<L_M^{(\beta)}\le 1$ implies $Z(\beta^\circ)\in(0,1]$.

Since $Z(\beta^\circ)>0$ and $Z(\beta)\to Z(\beta^\circ)$, there exists a neighborhood in which $Z(\beta)$ stays bounded away from zero; hence $T(\beta)$ is continuous at $\beta^\circ$. Because $\widehat{\beta}\xrightarrow{P}\beta^\circ$ and $T$ is continuous at $\beta^\circ$, the continuous mapping theorem gives $T(\widehat{\beta})\xrightarrow{P}T(\beta^\circ)$, i.e.,
\[
\int h(\theta)\,\pi_{M,\tau}^{(\widehat{\beta})}(\mathrm{d}\theta)
\ \xrightarrow{P}\
\int h(\theta)\,\pi_{M,\tau}^{(\beta^\circ)}(\mathrm{d}\theta).
\]
\end{proof}

\begin{proof}[Proof of Lemma~\ref{lem:inf-zero}.]
Since $\beta^\circ$ is obtained from an OLS projection with an intercept under $\mathcal L_{\mathrm{obs}}$, the normal equations include the intercept equation
$\mathbb{E}_{\mathrm{obs}}\!\big[Y-\beta^{\circ\top}\mathbf X\big]=0$. If $\mathcal L_{\mathrm{sim}}(\cdot\mid\theta^\star)=\mathcal L_{\mathrm{obs}}$, then
\[
\mu^\circ(\theta^\star)
=\mathbb{E}_{\mathrm{sim}(\theta^\star)}\!\big[Y-\beta^{\circ\top}\mathbf X\big]
=\mathbb{E}_{\mathrm{obs}}\!\big[Y-\beta^{\circ\top}\mathbf X\big]
=0.
\]
Hence $\mu^\circ(\theta^\star)^2=0$, which yields 
$\inf_{\vartheta\in\Theta}\mu^\circ(\vartheta)^2=0$ and $\theta^\star\in\Theta^\dagger$.
\end{proof}

\begin{proof}[Proof of Lemma~\ref{lem:be}]
Fix $\theta\in\Theta$
Let $Z_{m}$, $m=1,\dots,M$, be i.i.d.\ copies of $Z$ (conditional on $\theta$), and define
\[
R_{M,\theta} = \frac{1}{M}\sum_{m=1}^M Z_m\ ,
\qquad
S_{M,\theta} = \frac{\sqrt{M}\,(R_{M,\theta}-\mu^\circ(\theta))}{\sqrt{v^\circ(\theta)}}.
\]
Hence, $S_{M,\theta}$ has mean $0$ and variance $1$ (conditional on $\theta$). Set $\sigma(\theta)=\sqrt{v^\circ(\theta)}$ and $\widetilde Z:=Z-\mu^\circ(\theta)$, so that $\mathbb{E}[\widetilde Z\mid\theta]=0$,
$\Var(\widetilde Z\mid\theta)=v^\circ(\theta)$, and $S_{M,\theta}=M^{-1/2}\sum_{m=1}^M \widetilde Z_m/\sigma(\theta)$.
Let
\[
\rho_3(\theta) := \mathbb{E}\!\left[\left|\frac{\widetilde Z}{\sigma(\theta)}\right|^{3}\middle|\theta\right].
\]
We show $\sup_{\theta\in\Theta}\rho_3(\theta)<\infty$. By the triangle inequality and the elementary bound $|a-b|^{3}\le 4(|a|^{3}+|b|^{3})$,
\[
\mathbb{E}[|\widetilde Z|^{3}\mid\theta]
\ \le\ 4\Big(\mathbb{E}[|Z|^{3}\mid\theta]+|\mu^\circ(\theta)|^{3}\Big).
\]
Recall $\mathbb{E}[|Z|^{3}\mid\theta] \le C$ uniformly in $\theta$.
Continuity of $\mu^\circ$ on compact $\Theta$ implies $\sup_{\theta}|\mu^\circ(\theta)|<\infty$. Therefore,
\[
\sup_{\theta\in\Theta}\mathbb{E}[|\widetilde Z|^{3}\mid\theta]\ <\ \infty,
\qquad
\text{and hence}\qquad
\sup_{\theta\in\Theta}\rho_3(\theta)\ \le\ \frac{\sup_\theta \E[|\widetilde Z|^{3}\mid\theta]}{v_{\min}^{3/2}}\ <\ \infty.
\]
The (uniform) Berry--Esseen inequality for i.i.d.\ sums then yields a universal constant $C_{\mathrm{BE}}$ such that for all $t\in\mathbb{R}$ and all $\theta\in\Theta$,
\[
\left|\Pr(S_{M,\theta}\le t\mid\theta)-\Phi(t)\right|
\ \le\ \frac{C_{\mathrm{BE}}\,\rho_3(\theta)}{\sqrt{M}}
\ \le\ \frac{C'}{\sqrt{M}},
\]
where $C':=C_{\mathrm{BE}}\sup_{\theta}\rho_3(\theta)<\infty$. Finally, since
\[
\frac{R_{M,\theta}-\mu^\circ(\theta)}{\sqrt{v^\circ(\theta)/M}}
\ =\
S_{M,\theta},
\]
the display is exactly the asserted bound. Taking the supremum over $t\in\mathbb{R}$ and $\theta\in\Theta$ completes the proof.
\end{proof}

\begin{proof}[Proof of Lemma~\ref{lem:proxy}.]
Fix $\theta\in\Theta$ and abbreviate $\mu=\mu^\circ(\theta)$ and $s^2=v^\circ(\theta)/M$. Recall
\[
L_M(\theta;\tau)=\mathbb{E}\!\left[\exp\!\left\{-\frac{R_{M,\theta}^2}{2\tau^2}\right\}\middle|\theta\right],
\qquad
R_{M,\theta}=\frac{1}{M}\sum_{m=1}^{M}\!\big(Y_m^{\mathrm{sim}}-\beta^{\circ\top}\mathbf X_m^{\mathrm{sim}}\big).
\]
Let $F_\theta$ be the cdf of $R_{M,\theta}$ given $\theta$, and let $G_\theta$ denote the cdf of $\mathcal N(\mu,s^2)$. By Lemma~\ref{lem:be},
\begin{equation}
\label{eq:BE}
\Delta_M\ :=\ \sup_{t\in\R}\,|F_\theta(t)-G_\theta(t)|\ \le\ \frac{C'}{\sqrt{M}},
\end{equation}
with $C'$ independent of $\theta$.

To control the difference between expectations under the true and Gaussian laws,
consider the test function $g(r)=\exp(-r^2/(2\tau^2))$.  Since $g$ is absolutely
continuous with bounded derivative, we may express the difference in expectations
as a Stieltjes integral with respect to the distribution functions
$F_\theta$ and $G_\theta$ of $R_{M,\theta}$ and $\mathcal N(\mu,s^2)$ respectively:
\[
\mathbb{E}[g(R_{M,\theta})]-\mathbb{E}_{Z\sim\mathcal N(\mu,s^2)}[g(Z)]
\;=\;\int_{\R} g(r)\,\mathrm dF_\theta(r)-\int_{\R} g(r)\,\mathrm dG_\theta(r)
\;=\;\int_{\R} g(r)\,\mathrm d\!\big(F_\theta(r)-G_\theta(r)\big).
\]
By integration by parts for functions of bounded variation (see, e.g.,
Theorem~13.2.1 of Shorack and Wellner, 1986), for any two right-continuous,
nondecreasing functions $F,G$ with bounded variation and any
absolutely continuous $g$ with derivative $g'$ integrable on~$\R$,
\[
\Big|\int g\,\mathrm d(F-G)\Big| \ \le\ \operatorname{Var}(g)\,\sup_{t\in\R}|F(t)-G(t)|,
\]
where $\operatorname{Var}(g)=\int_{\R}|g'(r)|\,\mathrm dr$ denotes the total variation of $g$.  In our case, $F=F_\theta$ and $G=G_\theta$ are distribution functions, so the supremum $\sup_{t\in\R}|F(t)-G(t)|$ equals the Kolmogorov distance $\Delta_M$. Since $g'(r)=-(r/\tau^2)\exp(-r^2/(2\tau^2))$, its total variation is finite and easily computed:
\[
\operatorname{Var}(g)
=\int_{-\infty}^{\infty}\!|g'(r)|\,\mathrm dr
=\frac{1}{\tau^2}\int_{-\infty}^{\infty}\!|r|\,e^{-r^2/(2\tau^2)}\,\mathrm dr
=\frac{2}{\tau^2}\int_{0}^{\infty}\!r\,e^{-r^2/(2\tau^2)}\,\mathrm dr
=2.
\]
Combining this with the Berry--Esseen bound from Lemma~\ref{lem:be},
\[
\Delta_M
=\sup_t|F_\theta(t)-G_\theta(t)|
\le \frac{C'}{\sqrt{M}},
\]
yields the uniform inequality
\begin{equation}
\label{eq:BV}
\Big|\mathbb{E}[g(R_{M,\theta})]-\mathbb{E}_{Z\sim\mathcal N(\mu,s^2)}[g(Z)]\Big|
\ \le\ \operatorname{Var}(g)\,\Delta_M
\ \le\ 2\,\frac{C'}{\sqrt{M}},
\end{equation}
which holds for all $\theta\in\Theta$.  The bound is uniform in~$\theta$
because both $\operatorname{Var}(g)=2$ and the Berry--Esseen constant~$C'$ are independent of~$\theta$
under Assumption~\ref{ass:uniform-moments-fixed}.

It remains to compute the Gaussian expectation exactly. By Lemma~\ref{lem:L_M-gauss}, for $Z\sim\mathcal N(\mu,s^2)$,
\begin{align*}
\mathbb{E}\!\left[\exp\!\left\{-\frac{Z^2}{2\tau^2}\right\}\right] = \sqrt{\frac{\tau^2}{\tau^2+s^2}}\, \exp\!\left\{-\frac{\mu^2}{2(\tau^2+s^2)}\right\}.
\end{align*}
 Substituting $s^2=v^\circ(\theta)/M$ gives the leading term in Equation~\eqref{eq:proxy-main} with $\tau^2_{\mathrm{eff}}(\theta;M,\tau)=\tau^2+v^\circ(\theta)/M$. Finally, define
\[
r_{M,\tau}(\theta)\ :=\ \mathbb{E}[g(R_{M,\theta})]-\mathbb{E}_{Z\sim\mathcal N(\mu,s^2)}[g(Z)].
\]
The bound Equation~\eqref{eq:BV} shows that $\sup_{\theta\in\Theta}|r_{M,\tau}(\theta)|\le 2C'/\sqrt{M}$. Setting $C''=2C'$ completes the proof. The constants $C',C''$ are uniform in $\theta$ and do not depend on~$\tau$.
\end{proof}

\begin{proof}[Proof of Lemma~\ref{lem:laplace}.]
Let $m:=\inf_{\vartheta\in\Theta}\mu^\circ(\vartheta)^2$. Per Assumption~\ref{ass:exact_repres}, Lemma~\ref{lem:inf-zero} implies $m=0$. By Assumption~\ref{ass:ident-fixed}, for every $\theta\notin U$ we have
$\mu^\circ(\theta)^2\ge m+\eta=\eta$.
Let $s^2(\theta):=v^\circ(\theta)/M\in[v_{\min}/M,\,v_{\max}/M]$. By Lemma~\ref{lem:proxy},
\begin{equation}
\label{eq:proxy}
L_M(\theta;\tau)
=\sqrt{\frac{\tau^2}{\tau^2+s^2(\theta)}}\,
\exp\!\left\{-\frac{\mu^\circ(\theta)^2}{2(\tau^2+s^2(\theta))}\right\}
+r_{M,\tau}(\theta),
\qquad
\sup_{\theta\in\Theta}|r_{M,\tau}(\theta)|\le \frac{C''}{\sqrt{M}}.
\end{equation}

Numerator bound (outside $U$). Using $\sqrt{\frac{\tau^2}{\tau^2+s^2(\theta)}}\le 1$ and
$s^2(\theta)\le v_{\max}/M$,
\[
L_M(\theta;\tau)\ \le\ \exp\!\left\{-\frac{\mu^\circ(\theta)^2}{2(\tau^2+v_{\max}/M)}\right\}
+\frac{C''}{\sqrt{M}}
\ \le\ \exp\!\left\{-\frac{\eta}{2(\tau^2+v_{\max}/M)}\right\}
+\frac{C''}{\sqrt{M}}.
\]
Integrating,
\begin{equation}
\label{eq:num}
\int_{\Theta\setminus U} p(\theta)\,L_M(\theta;\tau)\,{\rm d}\theta
\ \le\
A\,\exp\!\left\{-\frac{\eta}{2(\tau^2+v_{\max}/M)}\right\}
+\frac{C''}{\sqrt{M}},
\end{equation}
where $A:=p_{\max}(\Theta)\lambda(\Theta)$, $p_{\max}(\Theta):=\sup_{\theta\in\Theta}p(\theta)$ and
$\lambda(\Theta)\in(0,\infty)$ is the Lebesgue measure of the compact set~$\Theta$.

Denominator bound (inside $U$). Because $U\supset\Theta^\dagger$ is open and $\mu^\circ$ is continuous, there exists a compact
$V \subset U$ with nonempty interior and, for any $\epsilon>0$,
\begin{equation}
\label{eq:V-def}
\mu^\circ(\theta)^2\ \le\ m+\epsilon=\epsilon
\qquad\forall\,\theta\in V.
\end{equation}
Fix
\[
\epsilon\ :=\ \frac{\tau^2+v_{\min}/M}{\tau^2+v_{\max}/M}\,\frac{\eta}{2}\ \le\ \frac{\eta}{2}.
\]
Using Equation~\eqref{eq:proxy}, $s^2(\theta)\in[v_{\min}/M,\,v_{\max}/M]$, and Equation~\eqref{eq:V-def},
\[
L_M(\theta;\tau)\ \ge\
\sqrt{\frac{\tau^2}{\tau^2+v_{\max}/M}}\,
\exp\!\left\{-\frac{\epsilon}{2(\tau^2+v_{\min}/M)}\right\}
-\frac{C''}{\sqrt{M}}
\qquad \forall\,\theta\in V.
\]
Since $p$ is continuous and strictly positive on a neighborhood of $\Theta^\dagger$
(Assumption~\ref{ass:ident-fixed}), there is $p_{\min}(B)>0$ with $p(\theta)\ge p_{\min}(V)$ on $V$. Let $\lambda(V)$ be the Lebesgue measure of the compact set~$V\subset U$.  Integrating over $V$ yields, with $B:=p_{\min}(V)\lambda(V)$,
\begin{equation}
\label{eq:den}
\int_{U} p(\theta)\,L_M(\theta;\tau)\,{\rm d}\theta
\ \ge\
B\,\sqrt{\frac{\tau^2}{\tau^2+v_{\max}/M}}\,
\exp\!\left\{-\frac{\epsilon}{2(\tau^2+v_{\min}/M)}\right\}
-\frac{C''}{\sqrt{M}}.
\end{equation}

By the choice of $\epsilon$,
\[
-\frac{\eta}{2(\tau^2+v_{\max}/M)}+\frac{\epsilon}{2(\tau^2+v_{\min}/M)}
\ =\ -\frac{\eta}{4(\tau^2+v_{\max}/M)}.
\]
For $M$ large enough and $M\tau^2\ge K$ (for some $K$ chosen below),
the right-hand side of Equation~\eqref{eq:den} is positive and the factor
$\sqrt{\frac{\tau^2}{\tau^2+v_{\max}/M}}$ is bounded below by a constant $c_K\in(0,1)$:
indeed, $M\tau^2\ge K$ implies
$\tau^2/(\tau^2+v_{\max}/M)\ge {K}/({K+v_{\max}})=:c_K^2$.
Combining Equations~\eqref{eq:num}--\eqref{eq:den} and factoring exponentials gives
\[
\frac{\int_{\Theta\setminus U} p L_M}{\int_{U} p L_M}
\ \le\
\frac{A}{B\,c_K}\,
\exp\!\left\{-\frac{\eta}{4(\tau^2+v_{\max}/M)}\right\}\,
\frac{1+\frac{C''}{\sqrt{M}A}\,e^{\frac{\eta}{2(\tau^2+v_{\max}/M)}}}
     {1-\frac{C''}{\sqrt{M}B c_K}\,e^{\frac{\epsilon}{2(\tau^2+v_{\min}/M)}}}.
\]
For $M\ge M_0$ (large enough that the denominator of the last fraction exceeds $1/2$),
the last fraction is bounded by $1+C'''/\sqrt{M}$ for some $C'''$ depending only on
$C'',A,B,c_K,\eta,v_{\min},v_{\max}$, and we may set
$C(U):=(A/(B\,c_K))\cdot 2$ to conclude
\[
\frac{\int_{\Theta\setminus U} p L_M}{\int_{U} p L_M}
\ \le\
C(U)\,\exp\!\left\{-\frac{\eta}{4(\tau^2+v_{\max}/M)}\right\}
\Bigl(1+\frac{C'''}{\sqrt{M}}\Bigr).
\]
All constants are independent of $(M,\tau)$ once $K$ and $M_0$ are fixed.
\end{proof}

\begin{proof}[Proof of Theorem~\ref{thm:concentration-fixed}.]
By moment representability and Lemma~\ref{lem:inf-zero}, $m:=\inf_{\vartheta\in\Theta}\mu^\circ(\vartheta)^2=0$.
Fix an open $U\supset\Theta^\dagger$ and let $\eta=\eta(U)>0$ be the separation constant from
Assumption~\ref{ass:ident-fixed}. By the population Laplace bound,
Lemma~\ref{lem:laplace}, there exist constants $C(U),C'''\!<\infty$ and $M_0,K>0$
(independent of $M,\tau$) such that for all $M\ge M_0$ and all $\tau>0$ with
$M\tau^2\ge K$,
\begin{equation}
\label{eq:key-ratio}
\frac{\int_{\Theta\setminus U} p(\theta)L_M(\theta;\tau)\,{\rm d}\theta}
     {\int_{U} p(\theta)L_M(\theta;\tau)\,{\rm d}\theta}
\ \le\
C(U)\,\exp\!\left\{-\frac{\eta}{4(\tau^2+v_{\max}/M)}\right\}
\Bigl(1+\frac{C'''}{\sqrt{M}}\Bigr).
\end{equation}
Condition~\eqref{eq:conc-rate} implies that for all large $M$, $M\tau_M^2\ge K$, hence Equation~\eqref{eq:key-ratio} applies with $\tau=\tau_M$ and the right-hand side tends to $0$ as
$M\to\infty$ and $\tau_M\downarrow 0$. Since
\[
\pi_{M,\tau_M}(U)
=\frac{\int_{U} p L_M}{\int_{\Theta} p L_M}
=\frac{1}{1+\frac{\int_{\Theta\setminus U} p L_M}{\int_{U} p L_M}},
\]
we obtain $\pi_{M,\tau_M}(U)\to 1$. For a closed $F$ with $F\cap\Theta^\dagger=\emptyset$, choose
an open $U$ with $\Theta^\dagger\subset U\subset\overline U\subset F^c$; then
$\pi_{M,\tau_M}(F)\le 1-\pi_{M,\tau_M}(U)\to 0$. This proves the population claims.

For the empirical measure, let $A$ be any Borel set (in particular $A=U$ or $A=F$).
By Theorem~\ref{thm:sn}, for each fixed $M$,
\[
\Pi_{n_\theta}(A)=\sum_{j=1}^{n_\theta} w_j\,\mathbbm{1}_A(\theta_j)
\ \xrightarrow[n_\theta\to\infty]{\text{a.s.}}\ \pi_{M,\tau_M}(A).
\]
Given $\varepsilon>0$, first take $M$ large enough that
$\pi_{M,\tau_M}(U)\ge 1-\varepsilon/2$ and $\pi_{M,\tau_M}(F)\le \varepsilon/2$; then take
$N(M)$ large enough that for all $n_\theta\ge N(M)$,
$|\Pi_{n_\theta}(A)-\pi_{M,\tau_M}(A)|\le \varepsilon/2$ almost surely.
Therefore, $\Pi_{n_\theta}(U)\ge 1-\varepsilon$ and $\Pi_{n_\theta}(F)\le \varepsilon$ almost surely.
Letting $\varepsilon\downarrow 0$ gives the claimed almost sure convergence.
\end{proof}

\begin{proof}[Proof of Proposition~\ref{prop:star-in-dagger}]
Let $\beta^\circ$ denote the best linear projection of $Y$ on the
augmented predictor $\mathbf X=(1,X^\top)^\top$ under the observed law
$\mathcal L_{\mathrm{obs}}(X,Y)$. OLS fitting with an
intercept yields
\[
\mathbb{E}_{\mathrm{obs}}\!\big[\mathbf X\,(Y-\beta^{\circ\top}\mathbf X)\big]=\mathbf 0.
\]
Taking the first (intercept) component gives
$\mathbb{E}_{\mathrm{obs}}[\,Y-\beta^{\circ\top}\mathbf X\,]=0$.

Fix $\theta^\star\in\Theta^\star$. Exact equality of laws,
$\mathcal L_{\mathrm{sim}}(X,Y\mid\theta^\star)=\mathcal L_{\mathrm{obs}}(X,Y)$, yields
$\mathbb{E}_{\mathrm{sim}}[\varphi(X,Y)\mid\theta^\star]=\mathbb{E}_{\mathrm{obs}}[\varphi(X,Y)]$ for all integrable $\varphi$. In particular,
\[
\mu^\circ(\theta^\star)
=\mathbb{E}_{\mathrm{sim}}\!\big[Y-\beta^{\circ\top}\mathbf X\mid\theta^\star\big]
=\mathbb{E}_{\mathrm{obs}}\!\big[Y-\beta^{\circ\top}\mathbf X\big]
=0.
\]
Hence $\mu^\circ(\theta^\star)^2$ attains the global minimum value $0$,
so $\theta^\star\in\arg\min_{\theta\in\Theta}\mu^\circ(\theta)^2=\Theta^\dagger$.
Therefore $\Theta^\star\subseteq\Theta^\dagger$.
\end{proof}

\bibliography{references}
\newpage

\end{document}